\documentclass[letterpaper,10pt,journal]{IEEEtran}

\usepackage[utf8]{inputenc}
\usepackage[centertags]{amsmath}
\usepackage{dsfont,amscd,amsthm,amsfonts,amsbsy,amssymb,mathrsfs}
\usepackage{bm}  
\usepackage{abraces}
\usepackage{graphicx} 
\usepackage{subfigure}
\usepackage{array}
\usepackage{caption}
\usepackage{epstopdf} 
\usepackage{bigstrut}
\usepackage{comment}
\usepackage{bigints}
\usepackage{empheq}
\usepackage{breqn}
\newcommand{\quotes}[1]{``#1''}

\usepackage[usenames, dvipsnames]{color}
\usepackage[dvipsnames]{xcolor}

\newtheorem{theorem}{Theorem}
\newtheorem{lemma}{Lemma}
\newtheorem{proposition}{Proposition}

\newcommand*{\Scale}[2][4]{\scalebox{#1}{$#2$}}

\begin{document}

\title{A Convex Optimization Approach for Backstepping PDE Design: Volterra and Fredholm Operators \\ }
\author{Pedro~Ascencio, Alessandro~Astolfi and Thomas~Parisini
\thanks{This work has been supported by the European Unions Horizon 2020 research and innovation programme under grant agreement No 739551 (KIOS CoE).}
\thanks{P. Ascencio is with the Warwick Manufacturing Group (WMG) at the University of Warwick, U.K., (\texttt{p.ascencio.o@gmail.com}).}%
\thanks{A. Astolfi is with the Department of Electrical and Electronic Engineering, Imperial College London, London, U.K., and also with Dipartimento di Informatica, Sistemi e Produzione, Universit\`{a} di Roma \quotes{Tor Vergata}, 00133 Rome, Italy (\texttt{a.astolfi@ic.ac.uk}).}%
\thanks{T. Parisini is with the Department of Electrical and Electronic Engineering, Imperial College London, London, U.K., with the KIOS Research and Innovation Centre of Excellence, University of Cyprus, and also with the Dept. of Engineering and Architecture at the University of Trieste, Italy (\texttt{t.parisini@gmail.com}).}%
}

\maketitle

\begin{abstract}
Backstepping design for boundary linear PDE is formulated as a convex optimization problem. Some classes of parabolic PDEs and a first-order hyperbolic PDE are studied, with particular attention to non-strict feedback structures. Based on the compactness of the Volterra and Fredholm-type operators involved, their Kernels are approximated via polynomial functions. The resulting Kernel-PDEs are optimized using Sum-of-Squares (SOS) decomposition and solved via semidefinite programming, with sufficient precision to guarantee the stability of the system in the $\mathcal{L}^2$-norm. The effectiveness and limitations of the approach proposed are illustrated by numerical solutions of some Kernel-PDEs.
\end{abstract}

\vspace{-0.2cm}

\section{INTRODUCTION}
In the field of Distributed Parameter Systems (DPSs), continuous-time Backstepping for linear Partial Differential Equations (PDEs) is a well-established  methodology in boundary control/observer design \cite{KS-2008,SK-2010,VK-2007,L-2009,M-2010}. Its fundamental idea, the \emph{Volterra transformation} \cite{RN-1990,Z-2012}, traces back to the application of the method of Integral Operators for solving initial-boundary problems \cite{C-1977} derived from the boundary control of parabolic equations \cite{S-1984}. It stands out for its elegant and simple systematic methodology, which: (i) does not involve spatial discretization of the PDE model (see \cite{R-1981} for fundamental disadvantages of early lumping), (ii) carries out a collective treatment of the system modes instead of a finite analysis of them based on their spectral characteristics (see \cite{C-2001} and references therein), (iii) does not require to formulate the problem in abstract Hilbert spaces, apply semigroup theory, nor solve operator-valued equations (see \cite{CZ-1995,BPDM-2007,TW-2009} for extension of classical control theory to infinite-dimensional systems).

Backstepping design for PDEs involves two main problems: (i) the solution/well-posedness of the so-called Kernel-PDE and (ii) the invertibility of the integral transformation. This methodology has mostly been applied to systems known as \quotes{\emph{strict-feedback}} systems on the basis of a Volterra-type transformation, invertibility of which is a well-know property \cite{L-2009, L-2003, GGK-2003}. It exploits the causal structure (causal in space \cite{C-2002}) leading to a kind of Kernel-PDE which is simple to solve in comparison with the operator Riccati equation derived from the linear quadratic regulator (LQR) approach \cite{CZ-1995,BDM-1976}. For some classes of these systems, the resulting Kernel-PDE can be reduced to a standard form, which allows obtaining a closed-form solution \cite{SK-2004,KS-2009}. For general cases
a closed-form analytic solution is hard to find and simple numeric methods cannot be applied directly \cite{M-2010}.

A common methodology used to solve the Kernel-PDEs (as well as to prove their well-posedness), consists in transforming these differential formulations into integral equations to be solved via the \emph{Successive Approximation} method. This way of solution has the objective to find a closed-form or provide a recursive computation of the integral Kernels \cite{KS-2008}-\cite{M-2010}, \cite{RN-1990,Z-2012,SK-2004,KS-2009,J-1999}. This kind of analysis is framed in the context of the Banach's contraction principle \cite{KK-2001}, tools also typically used to prove existence and uniqueness \cite{Z-2012,J-1999,K-2014,CMM-2006}. Since, for strict-feedback systems, this type of analysis has provided a useful and simple numerical tool, somewhat reduced research efforts have been devoted to solve the Kernel-PDE in alternative ways.

Recently, Backstepping for PDEs has been implemented on systems with \quotes{\emph{non-strict}} feedback structure on the basis of a \emph{Fredholm-type transformation}, for parabolic \cite{GXZ-2012,TBK-2014} as well as  hyperbolic PDEs \cite{BK-2014,BK-2015}. These kind of systems arise from multiple sources. For instance, naturally, in dynamics with non-local terms involving the whole spatial domain, in PDE models \cite{GXZ-2012} or in finite-dimensional systems with distributed delays \cite{BK-2011}. Additionally, in design-oriented problems such as: control of coupled PDE-ODE (Ordinary Differential Equations) systems by under-actuated schemes \cite{BK-2014,BK-2015} (fewer actuators than spatial states \cite{V-2006}; it avoids an additional control action to cancel the non-strict feedback term \cite{KS-2009}), or observer design for systems the output of which (sensing) comprises the states on the whole domain \cite{TH-2015}. In these cases, a Volterra-type transformation cannot be used (at least directly) and the application of a Fredholm-type transformation leads to new and intricate mathematical problems (operator invertibility, Kernel solvability) \cite{V-2006}. For instance, from the application of the concepts of fixed point theory (Picard sequence of successive approximations \cite{AAK-2014}) arises some system parameters constraints to guaranty the uniqueness of the Kernel-PDE solution, and thus the convergence of an approximate solution and the invertibility of a Fredholm-type transformation \cite{BK-2014,BK-2015}. This is due to the necessary condition of contraction of the resulting operator (Kernels with small spectral radius), one of the main drawbacks of this methodology of analysis for addressing general cases \cite{K-2014}. On the contrary, if a particular Kernel structure is proposed, such as partially separable Kernels, a simplified analysis can be carried out based on the the method of separation of variables \cite{TBK-2014}. However, under this approach the invertibility of the integral transformation and the solvability of the resulting Kernel-PDE are limited to a specific class of coefficients of the system.

In this article a novel methodology to solve approximately the Kernel-PDEs for both Volterra and Fredholm-type operators is presented. The proposed methodology recasts the Kernel-PDE as a convex optimization problem which: (i) obtains approximate Kernel solutions with sufficient precision to guarantee the stability of the closed-loop system and (ii) is not subject to the spectral characteristics of the resulting approximate operators. Assuming the well-posedness of the Kernel-PDEs, the main objective of the proposed approach is to determine Kernels to guarantee the stability of the system, which allows relaxing the exact zero matching condition on the differential boundary problem. Polynomial Kernels are proposed as approximate solution of the resulting Kernel-PDEs and the minimization of the residual functions is addressed by means of polynomial optimization tools. In particular, a Sum-of-Squares (SOS) decomposition problem is formulated -- equivalent to a convex optimization problem -- readily implementable resorting to semidefinite programming tools.

The paper is organized as follows. In Section~\ref{Background} the essential background, definitions and technical results are briefly introduced. In Section~\ref{Volterra} the problem of stabilization of parabolic PDE via the Volterra-type transformation is presented. In Section~\ref{Fredholm} the stabilization of hyperbolic PDE via the Fredholm-type transformation is analysed. In Section~\ref{Results} numerical results for specific examples related to Sections \ref{Volterra} and \ref{Fredholm} are presented. Concluding remarks are given in Section~\ref{Conclu}.

\vspace{-0.2cm}

\section{Preliminaries}
\label{Background}

\vspace{-0.05cm}

\subsection{Notation}
$\mathcal{A}(\Omega)$, $\mathcal{C}^r(\Omega)$ and $\mathcal{L}^2(\Omega)$ stand for the space of real-analytic functions, continuous functions with continuous first $r$ derivatives and square integrable functions on the domain $\Omega$, respectively. $\mathbb{I}$, $\mathbb{A}$, $\mathbb{V}$ and $\mathbb{F}$ denote the identity, integral, Volterra-type and Fredholm-type operator, respectively. $\mathbb{R}[\mathbf{x}]$ denotes the ring of real polynomials in $n$ variables $\mathbf{x}=[x_1,x_2,\ldots,x_n]^T$ and $\mathbb{P}[\mathbf{x}]=\{ p \in \mathbb{R}[\mathbf{x}]: p(x) \geq 0, \forall x \in \mathbb{R}^n \}$ stands for the set of non-negative real polynomials. The notation $\mathbb{R}_{n,r}[\mathbf{x}]$ and $\mathbb{P}_{n,r}[\mathbf{x}]$ explicitly
indicates polynomials in $n$ variables with degree at most $r$, whereas $\Sigma_s$ represents the subset of polynomials with Sum-of-Squares (SOS) decomposition. In particular $\mathbb{P}(\mathbf{K})$ represents the non-negative polynomials on the set $\mathbf{K}$.
$\Phi_r=[1,x_1,\ldots,x_n,x_1^2,x_1x_2,\ldots,x_n^r]^{\top}$ is the standard vector basis of $\mathbb{R}_{n,r}[\mathbf{x}]$. Polynomials are expressed by multi-index notation: $p(\mathbf{x})=\sum_{j=0}^{z(r)-1} p_j\mathbf{x}^{\overline{\alpha}_j}=\langle \overline{p}, \Phi_r \rangle$, where $r \in \mathbb{N}$ is the polynomial degree, $z(r)=\binom{n+r}{r}$ is the number of polynomial coefficients $\overline{p}=[p_0, p_1, \ldots, p_{z(r)-1}]^{\top} \in \mathbb{R}^{z(r)}$, $\mathbf{x}^{\overline{\alpha}_j}=x_1^{\alpha_j^1}\cdots x_n^{\alpha_j^n}$ represents the $j$-th monomial with powers $\overline{\alpha}_j=[\alpha_j^1,\ldots,\alpha_j^n]$ such that $|\overline{\alpha}_j|=\sum_{k=1}^{n} \alpha_j^k \leq r$, $\alpha_j^k \in \mathbb{N}$. An abstract form of this notation is given by $p(\mathbf{x})= \sum_{\bm{\alpha} \in \mathbb{N}_r^n} p_{\bm{\alpha}} \mathbf{x}^{\bm{\alpha}}$ with powers $\bm{\alpha} \in \mathbb{N}_r^n$, where $\mathbb{N}_r^n=\{\overline{\alpha}_j \in \mathbb{N}^n; |\overline{\alpha}_j|\leq r, \forall j=0,\ldots,z(r)-1\} =\{\overline{\alpha}_0,\ldots,\overline{\alpha}_{z(r)\!-\!1}\}$. $\mathbb{S}_{+}^m$ denotes the set of symmetric positive semidefinite matrices of dimension $m \times m$.

\vspace{-0.2cm}

\subsection{Integral Compact Operators}
Linear differential equations, ODEs and PDEs (boundary-value or initial-value problems), can be transformed into linear integral equations, the operators of which are frequently bounded or compact (completely continuous) \cite{K-2014,HS-1978,RY-2008}. In fact, every linear integral operator $\mathbb{A}: \mathcal{X} \rightarrow \mathcal{X}$:
\begin{align}
\label{int_op}
\mathbb{A}[u(\cdot)](\mathbf{x}) &:= \int_{\Omega \subset \mathbb{R}^n} K(\mathbf{x},\mathbf{y})u(\mathbf{y})d\mathbf{y}
\end{align}
with continuous Kernel or weakly singular Kernel $K$, is compact on the Banach space of continuous functions ($\mathcal{X}=(\mathcal{C}(\Omega;\mathbb{R}),\| \cdot \|_{\infty})$) and on the Hilbert space of square integrable functions ($\mathcal{X}=(\mathcal{L}^2(\Omega;\mathbb{R}), \langle \cdot, \cdot \rangle_{\mathcal{L}^2})$). Likewise, for square integrable Kernels, $\mathbb{A}$ is compact on this Hilbert space \cite{K-2014,ABK-2014}. This is the case for the integral operators derived from the Kernel-PDEs in the Backstepping PDE design, as pointed out in \cite{SK-2010} (page 19, footnote 2), where the Kernel is bounded and twice continuously differentiable.

In infinite-dimensional spaces bijectivity is a sufficient and necessary condition for (bounded) invertibility of bounded linear operators \cite{GGK-2003,RY-2008}. For linear equations of second kind; namely
\begin{align}
\label{ie_2kind}
(\mathbb{I}-\mathbb{A})[u(\cdot)](\mathbf{x})=w(\mathbf{x}),
\end{align}
two approaches are commonly carried out to determine whether there exists a bounded inverse. The first method is framed in the context of the \quotes{Banach contraction principle} \cite{KK-2001,AAK-2014}, based on the so-called Neumann series, for bounded operator with small spectral radius ($\|A\|<1$) \cite{K-2014}. This is the essential tool in the standard Backstepping PDE methodology, which relies on the inherent \emph{contraction} property of the Volterra operator \cite{KF-1957}, guaranteeing the uniform convergence of the Successive Approximation method \cite{RN-1990,Z-2012,J-1999,K-2014}. The second method is a re-statement of the celebrated Fredholm Alternative theorem \cite{K-2014,RY-2008}, based on the compactness property of $\mathbb{A}$, which is the core of the proposed approach. In this case the existence of a unique trivial solution $u=0$ of the homogeneous equation $u-\mathbb{A}u=0$ implies invertibility and thus the uniqueness of solutions.

Compact operators resemble the behaviour of operators in finite-dimensional spaces. In most of the traditional Banach spaces and for all Hilbert spaces, every compact operator is a limit of finite rank operators \cite{ABK-2014}. For continuous Kernels, a simple option for establishing this sequence are polynomials, which are a particular class of degenerate Kernels \cite{K-2014}.For instance if $K_N:=\sum_{j=0}^{N} k_j x^{\alpha_j} y^{\beta_j}$ then $\mathbb{A}_{K_N} [u(\cdot)](x)=\int_\Omega K_N(x,y) u(y) dy$ is a polynomial in the span of $\{x^{\alpha_j}\}_{j=0}^{N}$ so that $\mathbb{A}_{K_N}$ is a finite rank operator and hence compact \cite{ABK-2014}. Moreover, since $K \in \mathcal{C}(\Omega \times \Omega)$ ($\Omega=[0,1]$), based on the Weierstrass approximation theorem \cite{AH-2009}, there exists a sequence of polynomials $\{K_N\}_{N=0}^{\infty}$ such that:
\begin{flalign}
\label{k_pol}
\| \mathbb{A}_K \!-\! \mathbb{A}_{K_N} \| &=\sup_{\|u\|\leq 1}\|\mathbb{A}_{K-K_N} u\| \leq  \| K \!-\! K_N \|_\infty \underset{N \rightarrow \infty}{\rightarrow \quad 0}
\end{flalign}
with $k_j \in \mathbb{R}$, $\alpha_j \in \mathbb{N}$ and $\beta_j \in \mathbb{N}$ for $j=0,\ldots,N$. Equivalent results can be obtained for square integrable Kernels \cite{RN-1990,KS-2013}.

\vspace{-0.1cm}

\subsection{Polynomial Optimization: Sum-of-Squares}
In general, the global polynomial optimization problem:
\begin{flalign}
\label{pol_prob}
\mathbf{P}: \left\{ \ p^*= \inf_{\mathbf{x} \in \mathbf{K}} p(\mathbf{x}) \right. \Leftrightarrow \left\{ \begin{array}{rl}
p^*=& \ \sup \ \gamma \\ \text{subj. to:} & p(\mathbf{x})-\gamma\geq 0 \\
& \mathbf{x} \in \mathbf{K} \end{array}, \right.
\end{flalign}
where $\mathbf{K}:=\{\mathbf{x} \in \mathbb{R}^n; g_j(\mathbf{x}) \geq 0, j=1,\ldots,m\}$, $g_j \in \mathbb{R}[\mathbf{x}]$ and $p \in \mathbb{R}[\mathbf{x}]$, is NP-hard\footnote{The right hand side of \eqref{pol_prob} sets forth the dual formulation of $\mathbf{P}$, which cannot be solved in polynomial time for quartic or higher degree polynomials \cite{BPT-2013}. However, its non-negative constraints can be approximated, amongst others \cite{KP-2015}, via SOS, providing a convex formulation with computational tractable solution via semidefinite programming (interior point method, small-medium size problems).}. However, the problem  $\mathbf{P}$ can be efficiently approximated by a hierarchy of convex (semidefinite) relaxations $(\mathbf{P}_d; p^*_d)$, with $p^*_d$ global optimum for $\left(p(\mathbf{x})-\gamma \right) \in \mathbb{P}_{n,2d}[\mathbf{x}]$, using SOS representations for  non-negative
polynomials \cite{P-2000,P-2003, BPT-2013} or the theory of Moments \cite{L-2001,L-2010,L-2015}.

\begin{theorem}(\cite{BPT-2013,L-2010})
\label{SOS_theo}
Let $\Phi_r$ be the standard vector basis of $\mathbb{R}[\mathbf{x}]$ with $z(r)=\binom{n+r}{r}$ monomials in $\mathbf{x}$ with degree $\leq r$. A multivariate polynomial $p \in \mathbb{R}[\mathbf{x}]$ is SOS ($p \in \Sigma_s \subset \mathbb{P}_{n,2d}[\mathbf{x}])$ if and only if there exists a matrix $Q \in \mathbb{S}_+^{z(d)}$ satisfying $p(\mathbf{x})=\Phi_d^T Q \Phi_d$, $r=2d$.
\end{theorem}

For $\mathbf{K}$ \emph{compact basic semi-algebraic set}, the so-called \emph{Positivstellensatz} of Schm\"{u}dgen \cite{S-1991} and Putinar \cite{P-1993} allows formulating the hierarchy of semidefinite relaxations of \eqref{pol_prob} as:

\noindent
\begin{flalign}
\label{sos_relax}
\Scale[0.97]{
\mathbf{P}_d: \left\{ \! \! \!
\begin{array}{rl}
p^*_d = & \! \! \! \underset{Q_k} {\text{sup}} \  \gamma \\[-0.3cm]
\text{subj. to}: &\! \! \! \Big(p(\mathbf{x})\!-\!\gamma\!-\!\displaystyle\sum_{k=1}^{N} \overbrace{\Phi_{i_k}^{\top}(\mathbf{x}) Q_k \Phi_{i_k}(\mathbf{x})}^{s_k(\mathbf{x})} H_k(\mathbf{x}) \Big) \in \Sigma_s \\
  & Q_k \in \mathbb{S}_+^{z(i_k)}, \ i_k=d-d_k,  \ \forall k=1,\ldots,N,
  \end{array}
  \right.
}
\end{flalign}

\noindent
based on the existence of $s_k \in \Sigma_s$, where\footnote{The function $\lceil a \rceil$, commonly referred to as the \emph{ceil} function, rounds to the nearest integer greater than or equal to $a$.} $d_k=\lceil \text{deg}(H_k)/2 \rceil$ and $\max_k(\text{deg}(p),\text{deg}(H_k)) \leq 2d$ ; $N=2^m$ and $H_k=\prod_{k \in J} g_k$ for $J\subseteq\{1,\ldots,m\}$ if Schm\"udgen's Positivstellensatz is considered; $N=m$ and $H_k=g_k$ for Putinar's Positivstellensatz\footnote{Putinar's refinement requires that the quadratic module generated by $g_1,\ldots, g_m$ be Archimedean, which is not very restrictive \cite{L-2015}.}. Moreover, the approximate optimal solutions $p^*_k$ form a monotone nondecreasing sequence $(p^*_d \leq p^*_{d+1})$ such that $ p^*_d \rightarrow p^*$ as $d \rightarrow \infty$ \cite{ML-2010}.

\vspace{0.2cm}

From the dual viewpoint, in accordance with the Riesz linear functional $L_{\mathbf{y}} : \mathbb{P}(\mathbf{K}) \rightarrow \mathbb{R}$, $p  \mapsto L_{\mathbf{y}}(p)=\int_{\mathbf{K}} p(\mathbf{x}) d\mu(\mathbf{x})=\overline{p}^{\top} \mathbf{y}$ with $\mathbf{K}$ compact \cite{L-2010}, the proper cone of non-negative polynomials on $\mathbf{K}$ corresponds to the proper cone of \emph{moment sequences} $\mathbf{y}=(m_{\bm{\alpha}})_{\bm{\alpha} \in \mathbb{N}^n}$ with representing finite Borel measure $\mu$: $\mathbb{P}^*(\mathbf{K})=\{(m_{\bm{\alpha}}) \in \mathbb{R}^{\mathbb{N}^n}: \exists \ \mu \in \mathcal{M}(\mathbf{K})_+; \ m_{\bm{\alpha}} = \int_{\mathbf{K}} \mathbf{x}^{\bm{\alpha}} d\mu(\mathbf{x}), \ \forall  \bm{\alpha} \in \mathbb{N}^n \}$ \cite{ML-2010,L-2015}. For the full \textbf{K}-Moment problem a \quotes{practical}\footnote{General results such as the Riesz-Markov's \cite{RF-2010} and Riesz-Haviland's theorems \cite{L-2010}, do not have a \quotes{computationally tractable} characterization of the convex cone $\mathbb{P}^*(\mathbf{K})$.} necessary and sufficient condition for the existence of $\mu$ can be stated if a SOS representation on $\mathbf{K}$ compact basic semi-algebraic set is considered: $\exists \ \mu \in \mathcal{M}(\mathbf{K})_+ \Leftrightarrow \mathbf{M}_r(H_k \mathbf{y}) \succeq 0 \Leftrightarrow L_{\mathbf{y}}(H_k q^2) \geq 0 $, $\forall \ k=1,\ldots,N, \ \forall \ r \in \mathbb{N}$, with\footnote{For Schm\"udgen's Positivstellensatz: $N=2^m$ and $H_k=\prod_{k \in J} g_k$ for $J\subseteq\{1,\ldots,m\}$. For Putinar's Positivstellensatz: $N=m$ and $H_k=g_k$.}

\noindent
\begin{flalign}
& \Psi (H_k \mathbf{x})= \! \!\left[\sum_{\bm{\gamma}\in \mathbb{N}_t^n} (h_k)_{\bm{\gamma}} \mathbf{x}^{\bm{\gamma} + \overline{\alpha}_{0}}, \ldots, \sum_{\bm{\gamma} \in \mathbb{N}_t^n } (h_k)_{\bm{\gamma}} \mathbf{x}^{\bm{\gamma} + \overline{\alpha}_{z(r)-1}} \right]^{\top}, \nonumber
\end{flalign}
$q \in \mathbb{R}_{n,r}(\mathbf{K})$ and $H_k=\sum_{\bm{\gamma}\in \mathbb{N}_t^n} (h_k)_{\bm{\gamma}} \mathbf{x}^{\bm{\gamma}}$, $(h_k)_{\bm{\gamma}} \in \mathbb{R}$, where $\mathbf{M}_r(\mathbf{y})=L_{\mathbf{y}}(\phi_r(\mathbf{x}) \phi_r^{\top}(\mathbf{x}))$ is denominated \emph{Moment matrix} and $\mathbf{M}_r(H_k \mathbf{y})$ \emph{Localizing matrix}\footnote{Curto and Fialknow denominated the moment matrix of a shifted vector as \emph{localizing matrix} \cite{ML-2005}, which can be interpreted as the action to \emph{localize} the support of a representing measure of $\mathbf{y}$ \cite{ML-2010}.}.
Thus the tractable optimization problem \eqref{sos_relax} is a dual equivalent formulation of:
\begin{flalign}
\label{mom_relax}
\mathbf{D}_d: \left\{
\! \! \! \begin{array}{rl}
\rho_{d}^{*}= & \inf_{\mathbf{y}} L_{\mathbf{y}}(p)= \sum_{\bm{\alpha} \in \mathbb{N}_d^n} p_{\bm{\alpha}} m_{\bm{\alpha}} \\
\text{subj. to}:  & m_{\mathbf{0}}=1, \ \mathbf{M}_d(\mathbf{y}) \succeq 0 \\
 &  \mathbf{M}_{d-v_j}(g_j \mathbf{y}) \succeq 0, \ \forall j=1,\ldots,m,
 \end{array}
 \right.
\end{flalign}
where $\mathbf{y}$ is a moment \quotes{finite} sequence ($\bm{\alpha} \in \mathbb{N}_d^n$), $\mathbf{M}_d(\mathbf{y})$ is a Moment matrix, $\mathbf{M}_{d-v_j}(g_j \mathbf{y})$ is a Localizing matrix, $v_j\!:=\! \lceil \text{deg}(g_j)/2 \rceil$ and $d \geq \max \{ \lceil r/2 \rceil, \max_j \{v_j\} \}$.

\subsection{Convex Formulation of Differential BVPs}
Let $u=u(\mathbf{x}) \in \mathcal{C}^{w}(\Omega)$ be the solution of a linear PDE-Boundary Value Problem (BVP):

\noindent
\begin{flalign}
\label{PDE}
\mathbf{P}: \left\{ \begin{array}{rl}
\mathbb{L}[u(\cdot)](\mathbf{x}) &= f(\mathbf{x}), \  \forall \mathbf{x} \in \mathring{\Omega},\\
\mathbb{B}_i[u(\cdot)](\mathbf{x})|_{\mathbf{x} \in \partial\Omega_i} &= \mathbf{u}_i , \quad i=1,\ldots,r \in \mathbb{N}, \end{array} \right.
\end{flalign}

\noindent
in a compact domain $\Omega \subset \mathbb{R}^n$, where $\mathbb{L}$ and $\mathbb{B}_i$ are linear differential operators with polynomials terms, $f \in \mathbb{R}[\mathbf{x}]$, $\partial\Omega \supseteq \bigcup_{i=1}^r \partial\Omega_i$ represents the boundary of $\Omega=\mathring{\Omega} \ \cup \ \partial\Omega $, with $\mathring{\Omega}$ the interior of $\Omega$, and $\mathbf{u}_i \in \mathbb{R}$ is the value of $\mathbb{B}_i[u]$ on the boundary $\partial\Omega_i$. Based on the Weierstrass approximation theorem \cite{FS-2006}, the solution $u$ can be uniformly approximated by polynomials with theoretical arbitrary precision.

\noindent
Let

\noindent
\begin{flalign}
\label{Res_Fun}
\delta_d(\mathbf{x})=& \left(\sum_{\bm{\alpha} \in \mathbb{N}_d^n} p_{\bm{\alpha}} \mathbb{L}(\mathbf{x}^{\bm{\alpha}}) -  f(\mathbf{x})\right), \quad \forall \ \mathbf{x} \in \mathring{\Omega}
\end{flalign}

\noindent
be the \emph{residual function} due to the polynomial approximation $u(\mathbf{x})\approx \textsf{p}_d(\mathbf{x})=\sum_{\bm{\alpha} \in \mathbb{N}_d^n} p_{\bm{\alpha}} \mathbf{x}^{\bm{\alpha}} \in \mathbb{R}[\mathbf{x}]$ in \eqref{PDE}, of degree $d$ in $n$ variables $\mathbf{x}=[x_1,\ldots,x_n]^{\top} \in \mathbb{R}^{n}$ and coefficients $\bm{p}=[p_{\alpha_0},\ldots,$ $p_{\alpha_{z(n,d)-1}}]^{\top} \in \mathbb{R}^{z(n,d)}$.

\vspace{0.2cm}

\noindent
The main idea to solve \eqref{PDE} as a \emph{Polynomial Optimization Problem} is based on a notable result from Real Algebraic Geometry: the \emph{Positivstellensatz} \cite{BPT-2013}. Peculiarly, this result of \emph{positivity certification} does not depend on the characteristics of the polynomials involved in the problem. On the contrary, this result only relies on the kind of \emph{algebraic representation} of the domain $\Omega$. Thus, if $\Omega$ can be described by

\noindent
\begin{align}
\label{Set_O}
\Omega =& \{ \mathbf{x} \in \mathbb{R}^n; g_1(\mathbf{x})\geq 0, \ldots, g_m(\mathbf{x})\geq 0\},
\end{align}

\noindent
with $m \in \mathbb{N}$ and $g_j \in \mathbb{R}[\mathbf{x}]$, $\forall \ j=1,\ldots,m$, and this description is a \emph{compact basic semi-algebraic set}, based on the representation theorems of Schm{\"u}dgen or Putinar, the non-negative function $h(\mathbf{x})=\pm \delta_d(\mathbf{x})\mp \gamma$, for some $\gamma \geq 0$, can be formulated as:

\noindent
\begin{align}
\label{Res_Fun_SOS}
\begin{split}
& \left(\pm \delta_d(\mathbf{x})\mp \gamma - \sum_{J \neq 0} s_J(\mathbf{x}) G_J(\mathbf{x})\right) \in \Sigma_s,  \quad s_J \in \Sigma_s,
\end{split}
\end{align}

\noindent
$\forall \ \mathbf{x} \in \Omega$, where $G_J$ denotes a particular combination of polynomial constraints $g_j$'s in accordance with the Schm{\"u}dgen or Putinar representation selected. Thus, \eqref{Res_Fun_SOS} is a SOS decomposition problem equivalent to a convex optimization problem, numerically implementable via semidefinite programming \cite{BPT-2013,L-2015}.

\begin{flushleft}
\emph{Minimax Approximation}
\end{flushleft}

\vspace{-0.1cm}

Due to the \emph{axiom of completeness} and its consequence on the generalized \emph{Min-Max theorem} \cite{Dzung-2006}, the residual function \eqref{Res_Fun} is bounded by $\underline{\delta} \leq \delta_d(\mathbf{x}) \leq \overline{\delta}$, where $\underline{\delta}$ is the \emph{minimum} and $\overline{\delta}$ is the \emph{maximum} of $\delta_d$ on $\Omega$ compact domain. In addition, the exact solution of \eqref{PDE}, i.e. $\delta_d=0$ in \eqref{Res_Fun}, can be approximated using the simple idea of imposing $\underline{\delta} \rightarrow 0$ and $\overline{\delta} \rightarrow 0$ ($\delta_d \rightarrow 0$ in $\Omega$) on the extreme values of $\delta_d$. Intuitively, to achieve this objective, the optimization problem $\min \left\{\max_{x \in \Omega}\{|\overline{\delta}|,|\underline{\delta}|\} < \gamma \right\}$ can be formulated. This can be seen as the standard \emph{Uniform Best approximation} approach used to approximate the \emph{zero function} by $\delta_d$ in terms of $\mathcal{L}^{\infty}$-norm (uniform error) in $\Omega$, a scheme also denominated \emph{Minimax approximation} \cite{Powell-1981}.

\vspace{0.1cm}

\noindent
Similarly, a \emph{Least Squares approximation}, namely  $\min_{p_{\bm{\alpha}}} \int_{\Omega} \delta_d^2(\mathbf{x}) d\mathbf{x}$, can be carried out based on polynomial matrix inequalities and the Schur's complement.

\vspace{-0.1cm}

\subsection{Definitions and Technical Results}
For the representation of bivariate polynomial Kernels $P(x,y)=\sum_{k=0}^{z(d)-1}p_k x^{\alpha_k} y^{\beta_k}$ of degree $d \in \mathbb{N}$, coefficients $p_k \in \mathbb{R}$ and powers $\alpha_k \in \mathbb{N}$, $\beta_k \in \mathbb{N}$, the following standard basis of monomials is considered:
\begin{flalign}
\label{Def_pol}
& \Phi_d:=\{1,x,y,x^2,xy,y^2,x^3,x^2y,xy^2,y^3,\ldots \nonumber\\
&  \qquad \qquad \qquad \qquad \ldots,x^d,x^{d-1}y,\ldots,xy^{d-1},y^d\},
\end{flalign}
with $z(d)\!=\!(d\!+\!1)\!(d\!+\!2)/2$ terms ordered according to the monomials $\Phi_d(k)\!=\!x^{j-k} y^{k}$, $\forall \ k\!=\!0,\ldots,j$ and $j\!=\!0,\ldots,d$.

\begin{lemma}
\label{prop_K}
The Backstepping design for 1-dimensional PDEs with Volterra or Fredholm-type transformation involves the domains: $\Omega :=\{0\leq x \leq 1\}$, $\Omega_L :=\{0\leq x \leq 1, 0 \leq y \leq x\}$, $\Omega_U :=\{0\leq x \leq 1, x \leq y \leq 1\}$ \cite{KS-2008,BK-2015}, which can be formulated as:

\vspace{-0.1cm}

\noindent
\begin{flalign}
\label{Domains}
\Omega \!\equiv \!\{&x \in \mathbb{R}; g_1(x)\!=\!x(1-x)\!\geq\! 0\}, \nonumber \\
\Omega_L \!\equiv \! \{&(x,y) \in \mathbb{R}^2; g_1(x)\!\geq\! 0, g_2(x,y)\!=\!y(x-y)\!\geq\! 0\},\\
\Omega_U \!\equiv\! \{&(x,y) \in \mathbb{R}^2; g_1(x)\!\geq\! 0, g_3(x,y)\!=\!(y-x)(1-y)\!\geq 0\}. \nonumber
\end{flalign}

\vspace{-0.1cm}

\noindent
These domain representations are compact basic semi-algebraic sets and their associated quadratic modules are Archimedean.
\end{lemma}

\vspace{-0.15cm}

\begin{proof}
\label{proof_compact}
By definition the domains representations in \eqref{Domains} are basic semi-algebraic sets \cite{BPT-2013}. The equivalence of the sets $\Omega$, $\Omega_L$ and $\Omega_U$  and their respective representations in \eqref{Domains} is immediate, following the feasible solution set of the inequalities involved. Moreover, $\Omega$ is a closed and bounded subset of  $\mathbb{R}$ as well as the sets $\Omega_L$ and $\Omega_U$ are in $\mathbb{R}^2$, and hence compact \cite{WAS-1981}. The Archimedean property is verified since the quadratic modules associated to these sets satisfy \cite{ML-2010,L-2015}:

\vspace{-0.1cm}

\noindent
\begin{align}
 & \Scale[0.90]{\exists \ N \in \mathbb{N} \ \text{s.t.:} \ N-\sum_{i=1}^{n} x_i^2 \in \left\{ s_0+\sum_{j=1}^m s_j g_j; \ \left(s_j\right)_{j=0}^{m} \in \Sigma_s \right\}}. \nonumber
\end{align}
For instance, for $\Omega \equiv \{x \in \mathbb{R}, g_1(x)=x(1-x)\geq 0\}$, selecting $N=1$, $s_1=2 \geq 0$ and $s_0=(x-1)^2\geq 0$. For the lower triangular domain $\Omega_L \equiv \{(x,y) \in \mathbb{R}^2, g_1=x(1-x)\geq 0, g_2=y(x-y)\geq 0\}$, using:
\begin{align}
&N_L=2 \in \mathbb{N}, s_1=4 \geq 0, s_2=2 \geq 0, \nonumber \\
& s_0=3x^2-2xy-4x+y^2+2 = [1, x, y] Q_L [1, x, y]^T  \in \Sigma_s, \nonumber\\
&Q_L=
\Scale[0.90]{\left(\begin{array}{rrr} 2 & -2 & 0\\ -2 & 3 & -1 \\ 0 & -1 & 1
\end{array}\right)} \succeq 0 \Rightarrow \\ & \begin{array}{l} N_L-x^2-y^2=\\ s_0(x,y)+s_1g_1(x)+s_2g_2(x,y). \end{array} \nonumber
\end{align}
Likewise, in the case of the upper triangular domain $\Omega_U \equiv \{(x,y) \in \mathbb{R}^2, g_1(x)=x(1-x)\geq 0, g_3(x,y)=(y-x)(1-y)\geq0\}$, considering:
\begin{align}
&N_U=3 \in \mathbb{N}, s_1=4 \geq 0, s_2=2 \geq 0, \nonumber\\
&s_0=3x^2\!-\!2x\!-\!2y\!-\!2xy\!+\!y^2\!+\!3 = [1, x, y] Q_U [1, x, y]^T \in \Sigma_s, \nonumber\\
&Q_U=
\Scale[0.90]{\left(\begin{array}{rrr}
3 & -1 & -1\\ -1 & 3 & -1 \\ -1 & -1 & 1
\end{array}\right)} \succeq 0 \Rightarrow \\ & \begin{array}{l} N_U-x^2-y^2=\\ s_0(x,y)+s_1 g_1(x)+s_2 g_3(x,y).\end{array} \nonumber
\end{align}
\end{proof}

\vspace{-0.2cm}

\section{Parabolic PDE and the Volterra Operator}
\label{Volterra}

\subsection{Problem Setting}
In this section a class of parabolic PDEs with strict-feedback structure and spatially varying reactivity is considered \cite{KS-2008,SK-2010}:

\noindent
\begin{align}
\label{SP_1}
\begin{split}
u_t(x,t)&=\epsilon u_{xx}(x,t)+\lambda(x)u(x,t)\\
u(0,t)&=0, \quad u(1,t)=U(t),
\end{split}
\end{align}

\vspace{-0.05cm}

\noindent
where $u(x,t)=u_0(x) \in \mathcal{C}(\Omega;\mathbb{R})$ is the initial condition. The objective is to find a control action $U=U(t)$ so that the origin of \eqref{SP_1} is finite-time stable in the topology of the $\mathcal{L}^2$-norm. For this class of systems the Backstepping PDE methodology proposes a Volterra-type transformation (here referred to as Volterra operator):

\noindent
\begin{align}
\label{Vol_Op}
\begin{split}
w(x,t)&=u(x,t)-\int_0^x K(x,y)u(y,t)dy\\[-0.1cm]
&=(\mathbb{I}-\mathbb{V}_K)[u(\cdot,t)](x)
\end{split}
\end{align}
where $\mathbb{I}$ is the identity operator and $\mathbb{V}_K$: $\mathcal{C}(\Omega;\mathbb{R})\rightarrow \mathcal{C}(\Omega;\mathbb{R})$, to transform the system \eqref{SP_1} into the target stable system:

\noindent
\begin{align}
\label{TP_1}
\begin{split}
w_t(x,t)&=\epsilon w_{xx}(x,t)-c(x) w(x,t),\\
w(0,t)&=0, \quad w(1,t)=0,
\end{split}
\end{align}

\noindent
where $c(x)/\epsilon > - \pi^2/4$, $\forall \ x \in \Omega$, with a boundary feedback control determined by $U(t)=\int_0^1 K(1,y)u(y,t)dy$. Following the standard Backstepping PDE design procedure detailed in \cite{KS-2008}, the transformed system \eqref{SP_1} takes the form:

\noindent
\begin{flalign}
\label{P1_T}
& w_{t}(x,t)- \epsilon w_{xx}(x,t)+c(x) w(x,t) = \underbrace{\epsilon K(x,0)}_{\delta_0(x)} u_x(0,t) \ + \nonumber \\[-0.3cm]
&\underbrace{\left((\lambda(x)+c(x)) + 2 \epsilon \frac{d}{dx}K(x,x) \right)}_{\delta_1(x)} u(x,t) \ + \\[-0.2cm]
&\int_0^x \!\!\!\underbrace{\left(\epsilon K_{xx}(x,y) \!-\! \epsilon K_{yy}(x,y)\!-\!(\lambda(y)\!+\!c(y)) K(x,y) \right)}_{\delta_2(x,y)} u(y,t) dy, \qquad \nonumber
\end{flalign}

\noindent
so that the target system \eqref{TP_1} is achievable if the continuous bounded Kernel $K=K(x,y)$ satisfies the so-called Kernel-PDE:

\noindent
\begin{flalign}
\label{P1_ker0}
\delta_2(x,y)&=0, \quad \delta_1(x)=0, \quad \delta_0(x) = 0,
\end{flalign}

\noindent
$\forall \ (x,y) \in \Omega_L$ where the $\delta_i$ are defined in \eqref{P1_T}. In this article $\delta_i$ are denominated as \quotes{residual functions}. This linear hyperbolic PDE (Klein-Gordon-type) is well-posed and, for constant reactivity terms $\lambda=\lambda_0$ and $c=c_0$, it can be solved in closed-form \cite{KS-2008,SK-2010,SK-2004}:

\noindent
\begin{flalign}
\label{ker_closed}
K^\star(x,y)=-\overline{\lambda} y I_{1}(\sqrt{\Theta})/\sqrt{\Theta},
\end{flalign}
in terms of the first-order modified Bessel function $I_1$ with $\overline{\lambda}=(\lambda_0+c_0)/\epsilon$ and $\Theta=\overline{\lambda}(x^2-y^2)$.

\nopagebreak

\subsection{Kernel-PDE as a Convex Optimization Problem}
\begin{proposition}
\label{Prop_1}
Let

\noindent
\begin{flalign}
\label{Inv_Vol_Op}
u(x,t)&=w(x,t)+\int_0^x L(x,y)w(y,t)dy
\end{flalign}

\noindent
be the inverse transformation of \eqref{Vol_Op} \cite{KS-2008,L-2003} and:
\begin{flalign}
\label{L_decomp}
L(x,y)&=\breve{L}(x,y)-\sigma, \ \breve{L}(x,y)\geq 0, \ \sigma \geq 0
\end{flalign}

\noindent
a positive decomposition of $L$ in the triangular domain $\Omega_L$. Let $m_{0,0}$ be the 0-order moment of $\breve{L}$ in accordance with

\noindent
\begin{flalign}
\label{L_mom}
m_{i,j}&:=\int_0^1\int_0^x x^i y^j \breve{L}(x,y) dy dx.
\end{flalign}

\noindent
Let $c(x) \geq \underline{c} >  -\epsilon \pi^2/4$, $\forall \ x \in \Omega$,  $\overline{\delta_1}=\max_{x \in \Omega} |\delta_1(x)|$, $\overline{\delta_2}=\max_{(x,y) \in \Omega_L} |\delta_2(x,y)|$ and $\delta_0(x)=0$, the transformed system \eqref{P1_T} is exponentially stable in the $\mathcal{L}^2$-norm topology if the residual functions satisfy:

\noindent

\begin{align}
\label{Res_bound}
&\overline{\delta_1}+\overline{\delta_2} \leq \min \left\{\frac{\epsilon \theta (\frac{\pi^2}{4}+1)+(\underline{c}-\epsilon)}{(1+\sigma)},\frac{\epsilon(1-\theta)}{\left(m_{0,0}+\frac{\sigma+1/2}{4}\right)}\right\},
\end{align}

\noindent
for some scalar $0\leq \theta \leq 1$.

\end{proposition}
\begin{proof}
\label{proof_lyap}
Let $V=\frac{1}{2}\int_0^1w^2(x,t)dx=(1/2)\|w(x)\|^2$ be a Lyapunov functional\footnote{For a clearer description, the time-dependence in the functions is dropped ($w(x)\equiv w(x,t)$). The norm is the usual in the space of square integrable functions on the domain $\Omega=[0,1]$ : $\|w(x)\|^2=\int_{\Omega} w^2(x) dx$.}. Its time-derivative $\dot{V}=\int_0^1 w(x) w_t(x)dx$ along the trajectory (\ref{P1_T}), with $\delta_0(x)=0$, is given by:
\begin{flalign}
\label{Lyap_der}
\dot V &= \underbrace{\epsilon \int_0^1 \! w(x)w_{xx}(x)dx}_{T_1} + \! \underbrace{\int_0^1 \! w(x) u(x)\delta_1(x) dx}_{T_2} \\
& -\int_0^1 \! c(x) w^2(x) dx  \! + \! \underbrace{\int_0^1 w(x) \int_0^x \delta_2(x,y) u(y) dy dx}_{T_3}.\nonumber
\end{flalign}

\noindent
With respect to the term $T_1$, using integration by parts, the boundary conditions in \eqref{TP_1}, and splitting the resulting expression by a factor $0 \leq \theta \leq 1$, yields:
\begin{align}
\label{Lyap_der_b1_1}
\mbox{\small $T_1$} = & w(x)w_{x}(x)\left|_0^1 \right. -\int_0^1 w^2_x(x) dx, \nonumber \\
\leq & \underbrace{-\epsilon \theta \|w_x(x)\|^2}_{T_{1a}} -\underbrace{\epsilon (1-\theta) \|w_x(x)\|^2}_{T_{1b}}.
\end{align}
Then, applying the Wirtinger's inequality on the term $T_{1a}$ and the Agmon's and Young's inequalities on the term $T_{1b}$ \cite{KS-2008,SK-2010,HLP-1952,B-1969,W-1994}, yields:
\begin{align}
\label{Lyap_der_b1}
\mbox{\small $T_1$} \leq & -\epsilon \theta \frac{\pi^2}{4} \|w(x)\|^2 + \epsilon (1-\theta) ( \|w(x)\|^2- \overline{w}^2),
\end{align}
where $\overline{w}^2=\max_{x \in \Omega} w^2(x)$, which leads to:
\begin{align}
\label{Lyap_der_b2}
     \dot V \leq &  -\left(\epsilon\theta\left(\frac{\pi^2}{2}\!+\!2\right)\!+\! 2(\underline{c}\!-\!\epsilon)\right)V(t)\!-\!\epsilon(1\!-\!\theta)\overline{w}^2 + \mbox{\small $T_2 + T_3 $}.
\end{align}

\noindent
with $\underline{c}=\min_{x \in \Omega} c(x)$.  As for the term $T_2$, substituting $u=u(x,t)$ from \eqref{Inv_Vol_Op}, the inverse Kernel $L$ from \eqref{L_decomp}, taking an upper bound by means of the maximum absolute value of some integrand functions and the maximum value of the residual functions, yields:
\begin{align}
& \mbox{\small $T_2$} \!\! \leq \!\! \int_0^1 \!\!\! w^2(x) |\delta_1(x)| dx \!+ \!\! \int_0^1 \!\!\! |w(x)||\delta_1(x)|\int_0^x \!\!\! |L(x,y)||w(y)|dydx \nonumber\\
&\leq  \mbox{\small $\overline{\delta_1} \|w(x)\|^2 + \overline{\delta_1} \overline{w}^2 \int_0^x \breve{L}(x,y)dydx \! + \! \sigma \overline{\delta_1} \underbrace{\!\!\! \int_0^1 \!\!\! |w(x)| \int_0^x \!\!\! |w(y)| dy dx}_{T_{2a}}$}, \nonumber
\end{align}

\noindent
where $\overline{\delta_1}=\max_{x \in \Omega} |\delta_1(x)|$. Then, using the identity $\int_a^b f(x) \int_a^x f(y) dy dx = (1/2)(\int_a^b f(x)dx)^2$ for $f$ continuous function \cite{Fitz-2009} on the term $T_{2a}$, and the Gr\"{u}ss' Integral inequality \cite{D-1998} on its resulting term, leads to:
\begin{align}
\label{Lyap_der_b4}
\mbox{\small $T_2$} & \leq (2+\sigma)\overline{\delta_1}V(t) +  \overline{\delta_1}\left(\frac{\sigma}{8}+m_{0,0}\right)\overline{w}^2.
\end{align}
Changing the order of integration in the term $T_3$ and following a similar procedure as described above for the term $T_2$, yields:
\begin{align}
\label{Lyap_der_b5}
\mbox{\small $T_3$} &= \int_0^1 u(y) \int_y^1 w(x) \delta_2(x,y) dxdy \nonumber\\
& \leq \int_0^1|w(x)| \int_0^x|w(y)||\delta_2(x,y)|dy dx \nonumber \\
& +\int_0^1 \int_0^y\breve{L}(y,s)|w(s)|ds \int_y^1|w(x)||\delta_2(x,y)|dxdy \nonumber \\
& +\sigma \int_0^1\int_0^y |w(s)| ds \int_y^1 |w(x)||\delta_2(x,y)|dxdy \nonumber \\
& \leq (1+2\sigma) \overline{\delta_2}V(t) + \overline{\delta_2}\left(\frac{1+2\sigma}{8} + m_{0,0}\right)\overline{w}^2,
\end{align}
with $\overline{\delta_2}=\max_{(x,y) \in \Omega_L} |\delta_2(x,y)|$. Finally, using in \eqref{Lyap_der_b2} the upper bounds \eqref{Lyap_der_b4} and \eqref{Lyap_der_b5} for the terms $T_2$ and $T_3$, respectively, grouping terms with respect to $V=\|w(x)\|^2/2$ and $\overline{w}^2$, the condition $\dot{V} \leq 0 $ is satisfied if:

\begin{align}
\label{Lyap_der_b6}
& -(\overline{\delta_1}+\overline{\delta_2}) + \frac{\overline{\delta_1}\sigma+\overline{\delta_2}}{2(1+\sigma)} + \frac{\epsilon \theta (\pi^2/4+1)+(\underline{c}-\epsilon)}{(1+\sigma)} \geq 0,
\end{align}

\noindent
and

\noindent
\begin{align}
\label{Lyap_der_b7}
& -(\overline{\delta_1}+\overline{\delta_2}) + \frac{\overline{\delta_1}(1+\sigma)}{(8m_{0,0}+1+ 2\sigma)}+\frac{\epsilon(1-\theta)}{\left(m_{0,0}+\frac{\sigma+1/2}{4}\right)} \geq 0,
\end{align}

\noindent
which leads to the expression \eqref{Res_bound}, i.e. a sufficient condition for exponential stability of \eqref{P1_T} in the $\mathcal{L}^2$-norm topology.
\end{proof}

\vspace{0.1cm}

\noindent
Motivated by the result of Proposition \ref{Prop_1}, which sets forth a margin of clearance in the stability of this transformed system, a relaxation of the exact zero matching condition for the residual functions $\delta_1$ and $\delta_2$ can be considered. It allows formulating an approximate solution for the Kernel-PDE \eqref{P1_ker0}.

\begin{proposition}
\label{Prop_2} Let $N(x,y)=\sum_{k=0}^{z(d)-1} n_k x^{\alpha_k} y^{\beta_k}$ be a polynomial approximation of $K$ of arbitrary even degree $d > d_r=\max\{d_{\lambda},d_c\} \in \mathbb{N}$ in accordance with \eqref{Def_pol}. Let $\delta_1=\delta_1(x)$ and $\delta_2=\delta_2(x,y)$ be the resulting residual functions according to \eqref{P1_T} with polynomial degrees $d_1=d$ and $d_2=2\left\lceil\frac{d+d_r+2}{2}\right\rceil$, respectively; let $\underline{\rho}_1$, $\underline{\rho}_2$, $\overline{\rho}_1$, and $\overline{\rho}_2$ be lower and upper bounds of these functions in $\Omega_L$; $\gamma_j\geq 0, \ j=1,\ldots,4$. For reactivity terms $\lambda=\lambda(x)$ and $c=c(x)$ described by polynomial functions of degree $d_{\lambda}$ and $d_c$, respectively, the Kernel-PDE \eqref{P1_ker0} can be formulated as the convex optimization problem:

\vspace{-0.1cm}

\noindent
\begin{flalign}
\label{PDE_opt_1}
& \underset{\gamma_j, N, s_j}{\text{minimize:}} \qquad \gamma_1+ \gamma_2 + \gamma_3 + \gamma_4 \\
&\text{subject to:} \nonumber \\
\label{PDE_opt_4}
& \left(\delta_1(x)-\underline{\rho}_1-s_1(x) \ g_1(x)\right) \in \Sigma_s, \\
\label{PDE_opt_5}
& \left(\overline{\rho}_1-\delta_1(x)-s_2(x) \ g_1(x)\right) \in \Sigma_s, \\
\label{PDE_opt_6}
& \left(\delta_2(x,y)-\underline{\rho}_2-[s_3(x,y) \ s_4(x,y)] \ g_L(x,y) \right) \in \Sigma_s, \\
\label{PDE_opt_7}
& \left(\overline{\rho}_2-\delta_2(x,y)-[s_5(x,y) \ s_6(x,y)] \ g_L(x,y) \right) \in \Sigma_s, \\
\label{PDE_opt_8}
& s_j \in \Sigma_s, \forall \ j=1,\ldots,6, \\
\label{PDE_opt_9}
& \begin{bmatrix} \ \gamma_1 & \underline{\rho}_1 \\ \ \underline{\rho}_1 & \gamma_1 \end{bmatrix} \succeq 0, \quad  \begin{bmatrix} \ \gamma_2 & \overline{\rho}_1 \\ \ \overline{\rho}_1 & \gamma_2 \end{bmatrix} \succeq 0, \\
\label{PDE_opt_10}
& \begin{bmatrix} \ \gamma_3 & \underline{\rho}_2 \\ \ \underline{\rho}_2 & \gamma_3 \end{bmatrix} \succeq 0, \quad  \begin{bmatrix} \ \gamma_4 & \overline{\rho}_2 \\ \ \overline{\rho}_2 & \gamma_4 \end{bmatrix} \succeq 0, \\
\label{PDE_opt_11}
& \epsilon N_{xx}(x,y) \!-\! \epsilon N_{yy}(x,y)\!-\!(\lambda(y)\!+\!c(y)) N(x,y)\! =\! \delta_2(x,y)\\[-0.1cm]
\label{PDE_opt_12}
& (\lambda(x)+c(x)) + 2 \epsilon N_x(x,x)= \delta_1(x), \\
\label{PDE_opt_13}
& \delta_0(x)=N(x,0)=0,
\end{flalign}

\noindent
$\forall \ (x,y) \in \Omega_L$, for some polynomials $s_1,s_2$ of degree $d_1-2$, $s_3,s_4,s_5,s_6$ of degree $d_2-2$ and $g_L=[g_1, g_2]^{\top}$. The optimal minimal bounds for the residual functions are:  $\overline{\delta_1}=\max\{\gamma_1,\gamma_2\}$ and $\overline{\delta_2}=\max\{\gamma_3,\gamma_4\}$.
\end{proposition}

\vspace{-0.1cm}

\begin{proof}
Since $N$ is a polynomial approximation of $K$, as which is indicated above, the residual functions in \eqref{P1_T} have a polynomial structure determined by \eqref{PDE_opt_11}-\eqref{PDE_opt_13}. In addition, the quadratic module associated to the representation of the domains $\Omega$ and $\Omega_L$ are Archimedean (see Lemma \ref{prop_K}). Based on Putinar's Positivstellensatz \cite{ML-2010,L-2015}, via the SOS decomposition \eqref{PDE_opt_4}-\eqref{PDE_opt_8}, the unknown extreme values of $\delta_1$ and $\delta_2$  ($\underline{\rho}_1,\overline{\rho}_1,\underline{\rho}_2,\overline{\rho}_2$) can be determined via a polynomial optimization problem, which is convex in terms of polynomial coefficients and solved via semidefinite programming \cite{BPT-2013}. The absolute value of these upper and lower bounds are given by means of \eqref{PDE_opt_9}-\eqref{PDE_opt_10}, so that the linear cost function \eqref{PDE_opt_1} yields $\delta_1 \rightarrow 0$, $\delta_2 \rightarrow 0$ in $\Omega$ and $\Omega_L$, respectively.
\end{proof}

\subsection{Approximate Inverse Transformation}
It is well-known that Volterra operators of the second kind \eqref{Vol_Op} with continuous kernel have a unique solution which is globally invertible \cite{L-2009,L-2003,K-2014}. According to the standard Backstepping PDE procedure \cite{KS-2008,SK-2010}, the inverse Kernel $L$ in \eqref{Inv_Vol_Op} is determined following the same approach that leads to the Kernel $K$, and computed via the Successive Approximation method. However, this procedure is not suitable for residual functions $\delta_1(x)\neq 0$, $\delta_2(x,y)\neq 0$. In this section, a methodology based on the Moment problem described in \cite{L-2010,BC-2006,HLM-2014} to find an approximate smooth solution of the inverse kernel $L$ in \eqref{Inv_Vol_Op} is proposed\footnote{Alternatively, an optimization problem similar to Proposition \ref{Prop_2} can be applied directly to \eqref{DI_kernels} to find an approximation of $L$. This is formulated in Section \ref{Inv_Fredholm} for the Fredholm-type operator. Instead of using \eqref{DI_kernels}, the standard Backstepping PDE approach computes the inverse kernel following the same approach that leads to the direct kernel $K$ \cite{KS-2008}.}.

\begin{proposition}
\label{Prop_3}
Let
\begin{flalign}
\label{DI_kernels}
L(x,y) -K(x,y) - \int_y^x K(x,s)L(s,y)ds & = 0
\end{flalign}
$(\mathbb{T}[L(\cdot,\cdot)](x,y)=0)$ $\forall \ (x,y) \in \Omega_L$ be the relation satisfied by any direct kernel and its inverse with the Volterra-type transformations \eqref{Vol_Op} and \eqref{Inv_Vol_Op}. Let \eqref{L_decomp} be the positive decomposition of $L$, where $\sigma$ is its minimum value in $\Omega_L$. Let $N(x,y)=\sum_{k=0}^{z(d)-1} n_k x^{\alpha_k} y^{\beta_k}$ be a known polynomial approximation of $K$ in accordance with \eqref{Def_pol}, a sequence of approximate moments $m_{i,j}$ can be determined by the convex optimization problem:

\noindent
\begin{flalign}
\label{Mom_opt_1}
&\underset{\gamma_j, \varrho_j, m_{i,j},\sigma}{\text{minimize:}} \qquad \gamma_1+ \gamma_2 + \varrho_2-\varrho_1 \\
\label{Mom_opt_2}
&\text{subject to:} \quad \rho_{s_{ij}} \geq -\gamma_1, \ \rho_{s_{ij}} \leq \gamma_2, \\
\label{Mom_opt_3} & \sum_{i=0}^{2d_m}\sum_{j=0}^{2d_m} m_{i,j} \!+\! \sigma \!-\!\varrho_1 \geq 0, \
\varrho_2\!-\! \sum_{i=0}^{2d_m}\sum_{j=0}^{2d_m} m_{i,j}\! -\! \sigma  \geq 0, \\
\label{Mom_opt_5}
& m_{i,j}-\frac{\sigma}{(j\!+\!1)(i\!+\!j\!+\!2)}-\int_0^1\int_0^x N(x,y) x^i y^j dy dx \nonumber \\ &-\sum_{k=0}^{z(d)-1}\frac{n_k}{i\!+\!\alpha_k\!+\!1} \bigg( m_{\beta_k,j}-m_{i\!+\!\alpha_k\!+\!\beta_k+1,j} \nonumber \\
& \left.-\frac{\sigma}{j\!+\!1}\left(\frac{1}{j\!+\!\beta_k\!+\!2}+\frac{1}{i\!+\!j+\!\alpha_k\!+\!\beta_k\!+\!3} \right) \right) = \rho_{s_{ij}},\\
\label{Mom_opt_6}
& M_0 \geq 0, \ M_1 \succeq 0,\ M_2 \succeq 0, \ldots, M_{d_m} \succeq 0, \\
\label{Mom_opt_7}
& R_0 \geq 0, \ R_1 \succeq 0,\ R_2 \succeq 0, \ldots, R_{d_m-1} \succeq 0, \\
\label{Mom_opt_8}
& S_0 \geq 0, \ S_1 \succeq 0,\ S_2 \succeq 0, \ldots, S_{d_m-1} \succeq 0, \\
\label{Mom_opt_10}
& M_0\!=\!m(I_0)\!=\![m_{0,0}], M_1\!=\!m(I_1)\!=\!\begin{bmatrix} M_0 & m_{1,0} & m_{0,1} \\
                                        m_{1,0} & m_{2,0} & m_{1,1} \\
                                        m_{0,1} & m_{1,1} & m_{0,2} \end{bmatrix}, \nonumber \\
                                        & M_2=m(I_2), \ldots, M_{d_m}=m(I_{d_m}), \\
\label{Mom_opt_11}
& R_r =m(I_r+g_1(1))-m(I_r+g_1(2)), \\
\label{Mom_opt_12}
& S_r =m(I_r+g_2(1))-m(I_r+g_2(2)), 
\end{flalign}

\vspace{-0.6cm}

\small
\begin{flalign}
\label{Mom_opt_14}
& \Scale[0.90]{I_r\!=\!\begin{bmatrix} \begin{bmatrix} \begin{bmatrix} [0,0] & 1,0 & 0,1 \\
                                        1,0 & 2,0 & 1,1 \\
                                        0,1 & 1,1 & 0,2 \end{bmatrix} & \begin{matrix} 2,0 & 1,1 & 0,2\\
                                                                                       3,0 & 2,1 & 1,2\\
                                                                                       2,1 & 1,2 & 0,3 \end{matrix} \\
        \begin{matrix} 2,0 & 3,0 & 2,1\\
                       1,1 & 2,1 & 1,2\\
                       0,2 & 1,2 & 0,3 \end{matrix} &
                       \begin{matrix}  4,0 & 3,1 & 2,2 \\
                                        3,1 & 2,2 & 1,3 \\
                                        2,2 & 1,3 & 0,4 \end{matrix}
                      \end{bmatrix} &  \begin{matrix} 3,0 & \ldots \\ 4,0 & \ldots \\ 3,1 & \ldots \\ 5,0 & \ldots \\ 4,1 & \ldots \\ 3,2 & \ldots \end{matrix} \\
                      \begin{matrix} & 3,0 & 4,0 & 3,1 & & 5,0 & 4,1 & 3,2 \\ & \vdots & \vdots & \vdots & &  \vdots & \vdots & \vdots \end{matrix} & \begin{matrix} 6,0 & \ldots \\ \vdots & \ddots \end{matrix}
                       \end{bmatrix} },
\end{flalign}
\normalsize

\vspace{-0.6cm}

\begin{flalign}
\label{Mom_opt_15}
& g_1=  [ (1,0) , (2,0)], \ g_2=[ (1,1) , (0,2)], \qquad \qquad \qquad 
\end{flalign}
where $M_r$ are moment matrices of dimension $z(r)\!=(r+1)(r+2)/2$, $R_r$ and $S_r$
are localizing matrices (their sequence is limited by $r \leq d_m-\lceil \max(deg(g_i))/2 \rceil$) associated to the compact basic semi-algebraic set description of $\Omega_L$ \cite{L-2010}, the entries of which are indexed by $I_r$ according to the order of the powers in the canonical polynomial basis $\Phi_r$ in \eqref{Def_pol}; $s_{ij}=i(2d_m+1)+j$, $\forall \ i=0,\ldots,2d_m-d-1$, $j=0,\ldots,2d_m$, for an arbitrary moment order $d_m \geq \left\lceil\frac{d+1}{2}\right\rceil$ and $\gamma_1 \geq 0$, $\gamma_2 \geq 0$, $\varrho_1 \geq 0$ and $\varrho_2 \geq 0$.
\end{proposition}

\begin{proof}
Since the zero function is the only function orthogonal to every element in an inner product space:
$\langle \mathbb{T}[L], v\rangle_{\mathcal{L}^2}\!=\!0$, $\forall v \in \mathcal{L}^2(\Omega_L)$ $\Leftrightarrow \mathbb{T}[L] = 0$ (weak formulation) \cite{O-2014}, and due to the fact that the canonical polynomial basis $\Phi_r(x,y)=\{\phi_i(x)\phi_j(y)=x^i y^j; i+j \leq r, \forall i, j, r \in \mathbb{N}\}$ generates a dense subset of the (separable) Hilbert space $\mathcal{L}^2(\Omega)$, an approximate solution of \eqref{DI_kernels} can be found by means of the set of linear equations $\langle \mathbb{T}[L], \phi_i(x)\phi_j(y)\rangle_{\mathcal{L}^2(\Omega_L)}=0$ (finite dimensional problem):

\noindent
\begin{flalign}
& \int_0^1 \! \! \! \phi_i(x) \int_0^x \! \! \! L(x,y) \phi_j(y) dy dx \! - \! \int_0^1 \! \! \! \phi_i(x) \int_0^x \! \! \! N(x,y) \phi_j(y) dy dx \nonumber \\
\label{DI_ki}
& -\int_0^1 \phi_i(x)\int_0^x N(x,y) \int_0^y L(y,s) \phi_j(s) ds dy dx=0,
\end{flalign}

\noindent
for $N$ a known polynomial approximation of $K$. Interchanging the order of integration in the last term of \eqref{DI_ki}, substituting $L$ by \eqref{L_decomp} and plugging in the polynomial series $N$ and the test functions $\phi_i(x)=x^i, \phi_j(y)=y^j$ yields:

\vspace{-0.1cm}

\noindent
\begin{flalign}
\label{DI_mom_1}
& \int_0^1\!\!\!\int_0^x \!\!\!\breve{L}(x,y)x^i y^j dy dx -\frac{\sigma}{(j+1)(i+j+2)} \nonumber\\
& - \int_0^1 \!\!\!\left( \int_0^y \!\!\!(\breve{L}(y,s)\!-\!\sigma) s^j ds \right) \!\!\!\ \sum_{k=0}^{z(d)-1} \!\!\!\ n_k y^{\beta_k} \left(\frac{1\!-\!y^{i\!+\!\alpha_k\!+\!1}}{i\!+\!\alpha_k\!+\!1}\right) dy \nonumber\\
& - \int_0^1\!\!\!\int_0^x \!\!\!N(x,y) x^i y^j dy dx,
\end{flalign}

\vspace{-0.1cm}

\noindent
$\forall i \in \mathbb{N}$, $j \in \mathbb{N}$, $(i\!+\!j)\leq r$. Thus, expanding the second term in \eqref{DI_mom_1}, applying definition \eqref{L_mom} and  keeping the integral of the third term for numerical computation ($N$ is known), the expression \eqref{Mom_opt_5} is obtained. This equation is written as a residual function $\rho_{s_{ij}}$, the upper and lower bounds of which are set forth by \eqref{Mom_opt_2} and included in the optimization index \eqref{Mom_opt_1}, so that $\rho_{s_{ij}} \rightarrow 0$ as required. Since $L$ is not necessarily positive, it is decomposed as in \eqref{L_decomp}, where $\breve{L} \geq 0$ is considered as a density function for a Borel measure $\mu$: $d\mu=\breve{L}dx$, $m_{i,j}=\int_{\Omega_L}x^iy^j d\mu$. Based on the Putinar's representation of non-negative polynomials in a compact basic semialgebraic set \cite{ML-2010,BC-2006} described by \eqref{Mom_opt_15} in terms of polynomial powers, $\mathbf{y}_r=(m_{i,j})_{(i+j)\leq r}$ is a sequence of moments if and only if the Moment matrix $M_r$ and the Localizing matrices $R_r$, and $S_r$ are positive semi-definite $\forall r \in \mathbb{N}$. However, in practice, a truncated Moment problem can be solved ($r=0,\ldots,d_m$) as is formulated by conditions \eqref{Mom_opt_6}-\eqref{Mom_opt_8}. Finally, via the cost function \eqref{Mom_opt_1} with respect to $\varrho_1$ and $\varrho_2$, and the conditions of minimum and maximum \eqref{Mom_opt_3}, a bounded optimization problem is enforced, the solution of which yields a \quotes{valid}\footnote{Proposition \ref{Prop_3} states conditions to obtain a sequence of real numbers which is a \quotes{valid} sequence of moments, i.e., it corresponds to a moments of a non-negative function \cite{BC-2006}.} sequence of approximate moments minimizing \eqref{Mom_opt_5}.
\end{proof}

\begin{proposition}
\label{Prop_4}
Let $\breve{L}_r=\sum_{k=0}^{z(r)-1} l_k x^{\alpha_k} y^{\beta_k}$ be a polynomial approximation of $\breve{L}$ in \eqref{L_decomp} in accordance with \eqref{Def_pol}. Given a sequence of approximate moments $\mathbf{y}_r=(m_{i,j})_{(i\!+\!j)\leq r}$, $\breve{L}$ can be approximated as $\breve{L}_r=\theta^{\top} \Phi_r(x,y)$ via the solution of the convex optimization problem:

\noindent
\begin{flalign}
\label{Mom_rec_1}
& \underset{\gamma_j\geq 0, \theta, s_j}{\text{minimize:}}   \qquad \gamma_1 + \gamma_2 \\
& \text{subject to:}  \nonumber \\
\label{Mom_rec_2}
& (\mathbf{M}_r \theta-\mathbf{y}_r) + \gamma_1 \geq 0, \quad
\gamma_2 -(\mathbf{M}_r \theta-\mathbf{y}_r) \geq 0, \\
\label{Mom_rec_4}
& \theta^{\top} \Phi_r(x,y)\!-\!s_0(x,y)\!-[s_1(x,y) \ s_2(x,y)] g(x,y) \!=\! 0, \\
\label{Mom_rec_5}
& s_0,s_1,s_2 \in \Sigma_s,
\end{flalign}

\noindent
where $\gamma_1 \geq 0$, $\gamma_2 \geq 0$, $g=[g_1, g_2]^{\top}$ with $g_1(x)=x(1-x)$ and $g_2(x,y)=y(x-y)$; polynomials $s_0$ of degree $r$ and $s_1,s_2$ of degree $r\!-\!2$;
$\theta=[l_0,l_1,\ldots,l_{z(r)-1}]^{\top}$ and $\mathbf{M}_r = \int_0^1\int_0^x \Phi_r(x,y) \Phi_r^{\top}(x,y) dy dx \in \mathbb{S}_{+}^{z(r)}$, the elements of which are determined by:

\noindent
\begin{flalign}
\label{Mom_rdef_1}
& \mathbf{M}_r(i,j)= \frac{1}{(\alpha_i\!+\!\alpha_j\!+\!\beta_i\!+\!\beta_j\!+\!2)(\beta_i\!+\!\beta_j\!+\!1)},
\end{flalign}

\noindent
where the coefficients $\alpha_i, \beta_i$ correspond to the powers of the canonical polynomial basis $\Phi_r$ \eqref{Def_pol}. In addition, the sequence of minimizers $\breve{L}_r^*$, $\forall \ r \in \mathbb{N}$,  is such that $\|\breve{L} - \breve{L}_r^*\|_{\mathcal{L}^2(\Omega_L)} \rightarrow 0$ as $r \rightarrow \infty$, and therefore $L(x,y)\approx \breve{L}_r^*(x,y)-\hat{\sigma}$.
\end{proposition}

\vspace{-0.1cm}

\begin{proof}
Let $\phi_k(x,y)=x^{\alpha_k}y^{\beta_k}$ be the $k$-th element of the vector of monomial basis $\Phi_r$ in \eqref{Def_pol}. The $(i, j)$-th element (row and column, respectively) of the Moment matrix is given by
$\mathbf{M}_r(i,j) = \int_0^1 \int_0^x \left(x^{\alpha_i}y^{\beta_i}\right)\left(x^{\alpha_j}y^{\beta_j}\right)dy dx$,
the value of which is indicated in \eqref{Mom_rdef_1}. Since $\breve{L}_r^2=\big(\sum_{k=0}^{z(r)-1} l_k \phi_k \big) \allowbreak \big(\sum_{t=0}^{z(r)-1}l_t \phi_t \big)=\theta^{\top} \Phi_r \Phi_r^{\top} \theta$, where $\theta=[l_0,l_1,\ldots,l_{z(r)-1}]^{\top}$, following \cite{HLM-2014}, the mean square error of the approximation of $\breve{L}$ can be upper bounded by $E(\breve{L}_r)\!=\!\|\breve{L} \!-\! \breve{L}_r\|_{\mathcal{L}^2(\Omega_L)}^2 \leq \theta^T \mathbf{M}_r \theta - 2 \theta^T \mathbf{y}_r \!:= \! 2 J(\theta)$, where the \quotes{known} valid sequence of approximate moments $\mathbf{y}_r=(\int_{\Omega_L}\phi_k(x,y)d\mu)_{\alpha_k+\beta_k\leq r}$ is ordered according to


%
%
%



\noindent
\begin{flalign}
\label{Mom_rdef_2}
& \mathbf{y}_r=[m_{0,0},\ldots,m_{z(r),0},m_{0,1},\ldots,m_{z(r),1}, \ldots \nonumber \\
& \qquad \ m_{0,z(r)},\ldots, m_{z(r),z(r)}]^{\top}.
\end{flalign}

\vspace{-0.1cm}

\noindent
Since $J$ is quadratic in terms of the unknown vector $\theta$, its minimizer satisfies $\frac{d}{d\theta} \left(\frac{1}{2}\theta^T \mathbf{M}_r \theta - \theta^T \mathbf{y}_r\right)=\theta(\mathbf{M}_r+\mathbf{M}_r^T)/2-\mathbf{y}_r^T=0$ ($\mathbf{M}_r \theta = \mathbf{y}_r$), which has a global minimizer since $\frac{d^2J(\theta)}{d\theta^2}=\mathbf{M}_r \in \mathbb{S}_{+}^{z(r)}$ \cite{BV-2004}. To avoid the matrix inversion in $\theta=\mathbf{M}_r^{-1}\mathbf{y}_r$ and instead of a SDP formulation of $J$ via the Schur complement \cite{HLM-2014}, the approach proposed defines a lower bound $\gamma_1$ and an upper bound $\gamma_2$ via \eqref{Mom_rec_2}, which by means of \eqref{Mom_rec_1} imposes  $\mathbf{M}_r \theta \rightarrow \mathbf{y}_r$
. Thus, since the quadratic module associated to the representation of $\Omega_L$ is Archimedean (see Lemma \ref{prop_K}), $\breve{L}_r=\theta^T\Phi_r(x,y) \geq 0$, $\forall \ (x,y) \in \Omega_L$, if this satisfies the Putinar's Positivstellensatz representation given by \eqref{Mom_rec_4}-\eqref{Mom_rec_5} for some polynomials $s_0,s_1$ and $s_2$ with SOS decomposition. A proof of $\|\breve{L} - \breve{L}_r^*\|_{\mathcal{L}^2(\Omega_L)} \rightarrow 0$ as $r \rightarrow \infty$ is found in \cite[Proposition 4]{HLM-2014}.
\end{proof}

\section{Hyperbolic PIDE and the Volterra-Fredholm Operator}
\label{Fredholm}

\subsection{Problem Setting}
In this section a class of first-order hyperbolic PIDEs (Partial Integral Differential Equations) with non-causal structure is considered (see details in \cite{BK-2014,BK-2015}); namely

\noindent
\begin{flalign}
u_t(x,t)&= u_{x}(x,t)+f(x)u(0,t)+\int_0^x h_1(x,y)u(y,t)dy \nonumber\\[-0.2cm]
\label{SH_1}
&  + \int_x^1 h_2(x,y) u(y,t)dy,\\
u(1,t)&=U(t), \nonumber
\end{flalign}
%
where $u(x,0)=u_0(x) \in \mathcal{C}(\Omega)$ is the initial condition and $f,h_1,h_2$ are real-valued continuous functions. The aim is to find a control action $U$ so that the origin of \eqref{SH_1} is finite-time stable in the topology of the $\mathcal{L}^2$-norm. For this class of system, \cite{BK-2014,BK-2015} (see also \cite{TBK-2014} for parabolic systems) have proposed a Fredholm-type transformation (here referred to as Fredholm operator), namely

\noindent
%
\begin{align}
\label{Fred_Op}
w(x,t)&=u(x,t)\!-\!\!\int_0^x P(x,y)u(y,t)dy\!-\!\!\int_x^1 Q(x,y)u(y,t)dy, \nonumber\\
&=(\mathbb{I}-\mathbb{F}_{P,Q})[u(\cdot,t)](x),
\end{align}
%
where $\mathbb{F}_{P,Q}$: $\mathcal{C}(\Omega;\mathbb{R})\rightarrow \mathcal{C}(\Omega;\mathbb{R})$ is a linear operator in terms of Kernels $P$ and $Q$ in the lower $\Omega_L$ and upper $\Omega_U$ triangular domain, respectively, to transform the original system \eqref{SH_1} into the target stable system:
\begin{align}
\label{TH_1}
\begin{split}
w_t(x,t)&= w_{x}(x,t),\\
w(1,t)&=0,
\end{split}
\end{align}
with a boundary feedback control determined by $U(t)=\int_0^1 P(1,y)u(y,t)dy$. Following the standard Backstepping PDE design procedure (detailed in \cite{BK-2015}), the transformed system \eqref{SH_1} takes the form:
\begin{flalign}
\label{H1_T}
\begin{split}
& w_{t}(x,t)\!-\! w_{x}(x,t)\! = \! \delta_0(x) u(0,t) - \delta_3(x)u(1,t) \\ & + \int_0^x \!\! u(y,t) \delta_1(x,y) dy + \int_x^1 u(y,t) \delta_2(x,y) dy,
\end{split}
\end{flalign}
where the residual functions are:
\begin{flalign}
\label{HD0_T}
& \delta_0(x) =f(x)+P(x,0)-\int_0^x P(x,y)f(y)dy \nonumber\\[-0.2cm]
&            \quad \qquad -\int_x^1Q(x,y)f(y)dy, \quad \forall x \in [0,1], \\[-0.2cm]
\label{HD2_T}
& \delta_1(x,y)\!=\! h_1(x,y)\! +\! P_x(x,y)\! +\! P_y(x,y)\!-\!\!\int_0^y\!\!P(x,s)h_2(s,y)ds \nonumber \\[-0.2cm]
&  -\int_y^x\!\!P(x,s)h_1(s,y)ds \!- \!\!\int_x^1\!\!Q(x,s)h_1(s,y)ds, \ \forall (x,y) \in \Omega_L, \\[-0.2cm]
\label{HD3_T}
&\delta_2(x,y)\!=\! h_2(x,y)\! +\! Q_x(x,y)\!+\! Q_y(x,y)\! -\!\!\int_0^x\!\!P(x,s)h_2(s,y)ds \nonumber \\[-0.2cm]
&  -\int_x^y \!\!Q(x,s)h_2(s,y)ds\!-\!\!\int_y^1\!\!Q(x,s)h_1(s,y)ds, \ \forall (x,y) \in \Omega_U, \\[-0.1cm]
\label{HD1_T}
& \delta_3(x)=Q(x,1).
\end{flalign}
Thus, the target system \eqref{TH_1} is achievable if the continuous Kernels $P$ and $Q$ satisfy the so-called Kernel-PIDE\footnote{For these coupled hyperbolic PIDEs, a method of analysis, computation and an equivalent sufficient (conservative) condition for a unique solution have been given in \cite{BK-2015}.}:
\begin{align}
\label{H1_ker0}
\begin{split}
\delta_1(x,y)&=0, \qquad \forall \ (x,y) \in \Omega_L,\\
\delta_2(x,y)&=0, \qquad \forall \ (x,y) \in \Omega_U,\\
\delta_0(x) =0, \ \delta_3(x) &= 0, \qquad \forall \ x \in \Omega.
\end{split}
\end{align}
%

\subsection{Kernel-PIDE as a Convex Optimization Problem}
\begin{proposition}
\label{Prop_RBound_2}
Let

\noindent
\begin{flalign}
\label{Inv_Fred_Op}
& u(x,t)\!=\!w(x,t)\!+\!\!\int_0^x \! R(x,y)w(y,t)dy\!+\!\! \int_x^1 \! S(x,y)w(y,t)dy
\end{flalign}

\noindent
be the inverse transformation of \eqref{Fred_Op} in terms of the Kernels $R$ and $S$ (as it is proposed in \cite{TBK-2014,BK-2014,BK-2015} under specific conditions on the system \eqref{SH_1}). Let $\Delta_1=\int_0^1\int_0^x \delta_1^2 (x,y) dy dx$ and $\Delta_2=\int_0^1\int_x^1 \delta_2^2 (x,y) dy dx$ be the mean square of the residual functions \eqref{HD2_T} and \eqref{HD3_T}, respectively. Considering $\delta_0(x)=0$ and $\delta_3(x)=0$ in \eqref{HD0_T}-\eqref{HD1_T}, the transformed system \eqref{H1_T} is exponentially stable in $\mathcal{L}^2$-norm topology if the residual functions satisfy:

\noindent
\begin{align}
\label{Res_bound_2}
& \sqrt{\Delta_1} + \sqrt{\Delta_2} \leq \frac{e^{-1}}{\left(1+\sqrt{\sigma_1}+\sqrt{\sigma_2}\right)},
\end{align}

\noindent
where {\small $\sigma_1\!=\!\int_0^1\int_0^x R^2(x,y) dy dx$} and {\small $\sigma_2\!=\!\int_0^1\int_x^1 S^2 (x,y) dy dx$}.
\end{proposition}
\begin{proof}

\label{Proof_RBound_2}
Let $V=\frac{1}{2}\int_0^1 e^{\alpha x} w^2(x,t)dx$ be a Lyapunov functional for some $\alpha >0$. Its time-derivative $\dot{V}=\int_0^1 e^{\alpha x} w(x) w_t(x)dx$ along the trajectory \eqref{H1_T}, with $\delta_0(x)=0$ and $\delta_3(x)=0$, is given by:
\begin{flalign}
\label{RBound_der}
\dot V &\leq  \Scale[0.92]{\underbrace{\int_0^1 e^{\alpha x} w(x)w_{x}(x)dx}_{T_1}} \!+\! \Scale[0.92]{\underbrace{ \int_0^1 \! e^{\alpha x} w(x) \int_0^x u(y) \delta_1(x,y) dy dx}_{T_2}} \nonumber \\
 & \quad \Scale[0.92]{\underbrace{+  \int_0^1 \! e^{\alpha x} w(x) \int_x^1 u(y) \delta_2(x,y) dy dx}_{T_3}}.
\end{flalign}
Using integration by parts and the boundary condition in \eqref{TH_1} on the term $T_1$ yields:
\begin{flalign}
\label{RBound_der_1}
\begin{split}
T_1 &\leq  \Scale[0.95]{ \left.e^{\alpha x }w^2(x)\right|_{x=0}^{x=1}-\alpha/2 \int_0^1 e^{\alpha x} w^2(x) dx \leq -\alpha V(t)}.
\end{split}
\end{flalign}
Regarding the term $T_2$, changing the order of integration and plugging in $u=u(y,t)$ the expression given by the inverse transformation \eqref{Inv_Fred_Op}, this can be written as:

\begin{flalign}
\label{RBound_der_2}
T_2 &=  \Scale[0.92]{\underbrace{\int_0^1 w(y) \int_y^1 e^{\alpha x} w(x) \delta_1(x,y) dx dy}_{T_{2a}}}\\
 & \quad + \Scale[0.92]{\underbrace{ \int_0^1 \left(\int_0^y R(y,s)w(s) ds\right)\left(\int_y^1 e^{\alpha x} w(x) \delta_1(x,y) dx \right) dy }_{T_{2b}} } \nonumber \\
 & \quad + \Scale[0.92]{\underbrace{ \int_0^1 \left(\int_y^1 S(y,s) w(s) ds \right) \left(\int_y^1 e^{\alpha x} w(x) \delta_1(x,y) dx\right) dy}_{T_{2c}}}.\nonumber
\end{flalign}

\noindent
Using the Cauchy-Schwarz's integral inequality \cite{S-2004}, yields:

\noindent
\begin{flalign}
\label{RBound_der_3}
T_{2a}^2 & \leq \Scale[0.9999]{\left(\int_0^1 e^{2\alpha x} w^2(x) dx \right) \left(\int_0^1 \left(\int_0^x w(y) \delta_1(x,y) dy \right)^2 dx \right) }\nonumber \\
& \leq \Scale[0.9999]{\max\limits_{x \in \Omega}\{e^{\alpha x}\} \max\limits_{y \in \Omega}\{e^{-\alpha y}\} \ 4 V^2(t)} \ \cdot \nonumber \\
& \quad \Scale[0.9999]{\left( \int_0^1 \left(\int_0^x e^{\alpha y} w^2(y) dy \right) \left(\int_0^x \delta_1^2(x,y) dy \right) dx \right)} \Rightarrow \nonumber \\
T_{2a} & \leq 2 e^{\alpha/2} V(t) \sqrt{\Delta_1},
\end{flalign}

\noindent
\begin{flalign}
\label{RBound_der_4}
T_{2b}^2  & \leq \Scale[0.85]{\left(\int_0^1\!\left(\int_0^y R(y,s) w(s) ds\right)^2\! dy \right) \left(\int_0^1 \! \left(\int_y^1 e^{\alpha x} w(x) \delta_1(x,y) dx \right)^2 \! dy \right)} \nonumber \\
& \leq \Scale[0.9999]{\max\limits_{x \in \Omega}\{e^{\alpha x}\} \max\limits_{s \in \Omega}\{e^{-\alpha s}\} \left(\int_0^1 \left(\int_0^y R^2(y,s) ds \right) \cdot \right.} \nonumber \\
& \quad \Scale[0.9999]{ \left(\int_0^y e^{\alpha s} w^2(s) ds \right) dy \Big) \Big(\int_0^1 \left(\int_y^1 e^{\alpha x} w^2(x) dx \right) \cdot } \nonumber \\
& \quad \Scale[0.9999]{\left(\int_y^1 \delta_1^2(x,y) dx \right) dy \Big)} \Rightarrow \nonumber \\
T_{2b} & \leq 2 e^{\alpha/2} V(t) \sqrt{\Delta_1} \Scale[0.9999]{\left(\int_0^1\int_0^y R^2(y,s) ds dy\right)^{\frac{1}{2}}},
\end{flalign}

\noindent
and

\noindent
\begin{flalign}
\label{RBound_der_5}
T_{2c}^2 & \leq \Scale[0.85]{\left(\int_0^1\!\left(\int_y^1 S(y,s) w(s) ds\right)^2\!dy \right)}
\Scale[0.85]{\left(\int_0^1 \! \left(\int_y^1 e^{\alpha x} w(x) \delta_1(x,y) dx \right)^2\!dy \right)} \nonumber \\
& \leq 4 e^{\alpha} V^2(t) \Scale[0.85]{\left(\int_0^1 \left(\int_y^1 S^2(y,s) ds \right) dy \right)\left(\int_0^1 \left(\int_y^1 \delta_1^2(x,y) dx \right) dy \right)} \Rightarrow \nonumber \\
T_{2c} & \leq 2 e^{\alpha/2} V(t) \sqrt{\Delta_1} \Scale[0.9999]{\left(\int_0^1\int_y^1 S^2(y,s) ds dy\right)^{\frac{1}{2}}},
\end{flalign}

\noindent
where $\Delta_1=\int_0^1\int_0^x \delta_1^2(x,y) dydx$. With respect to the term $T_{3}$, this can be written as:

\noindent
\begin{flalign}
\label{RBound_der_6}
T_3 &=  \Scale[0.92]{\underbrace{ \int_0^1 w(y) \int_0^y e^{\alpha x} w(x) \delta_2(x,y) dx dy}_{T_{3a}}} \\
 & \quad + \Scale[0.92]{\underbrace{ \int_0^1 \left(\int_0^y R(y,s)w(s) ds\right)\left(\int_0^y e^{\alpha x} w(x) \delta_2(x,y) dx \right) dy }_{T_{3b}} } \nonumber \\
 & \quad + \Scale[0.92]{\underbrace{ \int_0^1 \left(\int_y^1 S(y,s) w(s) ds \right) \left(\int_0^y e^{\alpha x} w(x) \delta_2(x,y) dx\right) dy}_{T_{3c}}}, \nonumber
\end{flalign}

\noindent
where upper bounds for every term, in relation with $\Delta_2=\int_0^1\int_x^1 \delta_2^2(x,y) dydx$, can be found following the same procedure described above in  \eqref{RBound_der_3}-\eqref{RBound_der_5}. Finally, using in \eqref{RBound_der} the upper bounds for the terms $T_1$, $T_2$ and $T_3$, the condition $\dot{V} \leq 0 $ is satisfied provided:
\begin{flalign}
& \alpha - 2 e^{\alpha/2} \left(\sqrt{\Delta_1}+\sqrt{\Delta_2} \right) \left(1 + \sqrt{\sigma_1} + \sqrt{\sigma_2} \right) \geq 0, \nonumber
\end{flalign}
with {\small $\sigma_1=\int_0^1 \int_0^x R^2(x,y) dy dx$} and {\small$\sigma_2=\int_0^1 \int_x^1 S^2(x,y) dy dx$}. Since the factor $(1/2) \alpha e^{-\alpha/2}$ reaches the maximum value of $e^{-1}$ at $\alpha=2$, the expression \eqref{Res_bound_2} is obtained, i.e. a sufficient condition for exponential stability of \eqref{H1_T} in the $\mathcal{L}^2$-norm topology.

\end{proof}

\vspace{0.1cm}

\noindent
Based on this result, similarly to Proposition \ref{Prop_2}, a relaxation on the zero matching condition for the residual functions $\delta_1$ and $\delta_2$ can be considered and the Kernel-PIDE \eqref{H1_ker0} can be solved approximately in terms of polynomial Kernels.

\begin{proposition}
\label{Prop_5} Let $N(x,y)=\sum_{k=0}^{z(d)-1} n_k x^{\alpha_k} y^{\beta_k}$ and $M(x,y)=\sum_{k=0}^{z(d)-1} m_k x^{\alpha_k} y^{\beta_k}$ be polynomial approximations of $P$ and $Q$, respectively, of arbitrary even degree $d \in \mathbb{N}$, with coefficients $n_k$ and $m_k$ and powers in accordance with \eqref{Def_pol}. Let $\delta_1=\delta_1(x,y)$ and $\delta_2=\delta_2(x,y)$ be the resulting residual functions according to \eqref{HD2_T} and \eqref{HD3_T}, respectively, with degree $d_{\delta}\!=\!2\left\lceil(\max\{d\!+\!d_{h_1},d\!+\!d_{h_2}\}\!+\!1)/2\right\rceil$ and $\gamma_1 \geq 0$, $\gamma_2 \geq 0$. For any functions $f,h_1,h_2$ described by polynomials of degree $d_f, d_{h_1}, d_{h_2}$, respectively, the Kernel-PIDE \eqref{H1_ker0} can be formulated as the convex optimization problem\footnote{The expression $\delta=\delta(\mathbf{x}) \left|_{Q \approx M}^{P \approx N} \right.$ indicates that in the function $\delta$, the Kernels $P$ and $Q$ has been substituted by the polynomials $N$ and $M$, respectively.}:

\noindent
\begin{flalign}
\label{PIDE_opt_1}
& \underset{\gamma_j, N, M, T_1, T_2, s_j}{\text{minimize:}} \qquad \gamma_1+ \gamma_2 \\
&\text{subject to:} \nonumber \\
\label{PIDE_opt_4}
& \Scale[0.92]{\begin{bmatrix} 2\gamma_1\!-\!T_1(x,y)\!-\!s_1(x,y)g_1(x) & \delta_1(x,y)\\ \delta_1(x,y) & \gamma_1\!-\!s_2(x,y)g_2(x,y) \end{bmatrix}} \in \Sigma_s^{2 \times 2}, \\
\label{PIDE_opt_5}
& \Scale[0.92]{\begin{bmatrix} 2\gamma_2\!-\!T_2(x,y)\!-\!s_3(x,y)g_1(x) & \delta_2(x,y) \\ \delta_2(x,y) & \gamma_2\!-s_4(x,y)g_3(x,y) \end{bmatrix}} \in \Sigma_s^{2 \times 2}, \\[-0.1cm]
\label{PIDE_opt_6}
& s_1,s_2,s_3,s_4\in \Sigma_s, \\
\label{PIDE_opt_7}
& \int_0^1\int_0^x T_1(x,y)dy dx = 0, \ \int_0^1\int_x^1 T_2(x,y)dy dx = 0\\
\label{PIDE_opt_8}
& \delta_1=\delta_1(x,y) \left|_{Q \approx M}^{P \approx N} \right. \qquad  \qquad \qquad \mbox{as in (\ref{HD2_T})}, \\
\label{PIDE_opt_9}
& \delta_2=\delta_2(x,y) \left|_{Q \approx M}^{P \approx N} \right. \qquad  \qquad \qquad \mbox{as in (\ref{HD3_T})},\\
\label{PIDE_opt_10}
& \delta_0=\delta_0(x) \left|_{Q \approx M}^{P \approx N} \right.= 0 \qquad  \qquad \ \quad \mbox{as in (\ref{HD0_T})}, \\
\label{PIDE_opt_11}
& \delta_3=\delta_3(x) \left|_{Q \approx M}\right.= M(x,1) = 0 \ \ \mbox{as in (\ref{HD1_T})},
\end{flalign}
for some polynomials $s_j, j=1,\ldots,4$ of degree $d_{\delta}-2$, $T_1$ and $T_2$ of degree $2d_{\delta}$, $g_1(x)=x(1-x)$, $g_2(x,y)=y(x-y)$ and $g_3(x,y)=(1-y)(y-x)$. The optimal root mean square bounds of the residual functions are: \small{$\sqrt{\Delta_1} \leq \gamma_1$} and \small{$\sqrt{\Delta_2} \leq \gamma_2$}.
\end{proposition}

\begin{proof}
The convex optimization problem formulation follows similar arguments as the ones given in the proof of Proposition \ref{Prop_2}. Regarding the optimal mean square bounds for $\delta_1$ and $\delta_2$, let
\begin{flalign}
\label{Pos_cond1}
& A_1= \begin{bmatrix} 2\gamma_1-T_1(x,y) & \delta_1(x,y)\\ \delta_1(x,y) & \gamma_1 \end{bmatrix} \succ 0, \forall (x,y) \in \Omega_L
\end{flalign}
be a symmetric real polynomial positive definite matrix on $\Omega_L$ (pointwise condition). Taking the Schur's complement of $A_1$, its integration on the domain $\Omega_L$ yields
\begin{flalign}
& A_1 \succ 0 \Leftrightarrow 2\gamma_1^2-\gamma_1 T_1(x,y)- \delta_1^2(x,y) > 0, \ \forall (x,y) \in \Omega_L, \nonumber \\
\label{Pos_cond2}
& \int\limits_0^1\int\limits_0^x \delta_1^2(x,y) dydx < \gamma_1^2 \!-\! \gamma_1\int\limits_0^1\int\limits_0^x T_1(x,y) dy dx.
\end{flalign}
If there exists a polynomial function $T_1$ satisfying \eqref{PIDE_opt_7}, it is clear that $\gamma_1$ is an upper bound of the root mean square value of $Y$ in $\Omega_L$. The positivity condition \eqref{Pos_cond1} is made computationally tractable via the matrix-polynomial version of Putinar's Positivstellensatz\footnote{This condition can also be "scalarized", i.e., expressed in terms of scalar polynomials \cite{BPT-2013,L-2015,SchH-2006}.} \cite{L-2015,SchH-2006}, based on the representation of $\Omega_L$ as in Lemma \ref{prop_K}, namely, for some $\rho >0 $, $A_1(x,y) \succ \rho I \succ 0, \ \forall (x,y)\in \Omega_L$, then:
\begin{flalign}
\label{SOS_Mcond1}
& \Big(A_1(x,y)\!-\!g_1(x)S_1(x,y)\!-\!g_2(x,y)S_2(x,y)\Big) \!\in \Sigma_s^{2 \times 2} \\
& \ S_1, \ S_2 \in \Sigma_s^{2 \times 2}.
\end{flalign}

\noindent
Thus, conditions \eqref{PIDE_opt_4} and \eqref{PIDE_opt_5} are obtained if the following particular forms for $S_1$ and $S_2$ are considered:
\begin{flalign}
S_1 &=\begin{bmatrix} s_1(x,y) & 0\\ 0 & 0 \end{bmatrix} \in \Sigma_s^{2 \times 2}, \quad s_1 \in \Sigma_s, \\
S_2 &= \begin{bmatrix} 0 & 0\\ 0 & s_2(x,y) \end{bmatrix} \in \Sigma_s^{2 \times 2}, \quad s_2 \in \Sigma_s,
\end{flalign}
the SOS matrix condition of which is immediately verified since $S_1=B B^T$, $S_2=C C^T$ with:
\begin{flalign}
B&= \begin{bmatrix} b_1(x,y) & \cdots & b_{m_1}(x,y) & 0 \\ 0 & \cdots & 0 & 0 \end{bmatrix} \in \mathbb{R}[\mathbf{x}]^{2 \times (m_1\!+\!1)}, \\
C&= \begin{bmatrix} 0 & 0 & \cdots & 0 \\ 0 & c_1(x,y)& \cdots & c_{m_2}(x,y) \end{bmatrix} \in \mathbb{R}[\mathbf{x}]^{2 \times (m_2\!+\!1)},
\end{flalign}
where $s_1=\sum_{j=1}^{m_1} b_j^2(x,y) \in \Sigma_s$ and $s_2=\sum_{j=1}^{m_2} c_j^2(x,y) \in \Sigma_s$; $b_j \in \mathbb{R}[\mathbf{x}], \forall j=1,\ldots,m_1$, $c_j \in \mathbb{R}[\mathbf{x}], \forall j=1,\ldots,m_2$, for some finite $m_1 \in \mathbb{N}$ and $m_2 \in \mathbb{N}$. Following the same arguments, conditions for $Z$ on $\Omega_U$ can be deduced. Therefore, according to the optimization objective \eqref{PIDE_opt_1}, the root mean square error of $\delta_1$ and $\delta_2$ are minimized.
\end{proof}

\subsection{Approximate Inverse Transformation}
\label{Inv_Fredholm}
For known Kernels $P$ and $Q$, the inverse transformation of \eqref{Fred_Op} can be found by means of the direct substitution of \eqref{Inv_Fred_Op} in \eqref{Fred_Op}, which yields

\noindent
\begin{flalign}
\label{DI_Fred_0}
& \int_0^x w(y,t) \delta_1(x,y) dy + \int_x^1 w(y,t) \delta_2(x,y) dy =0,
\end{flalign}
with the equality satisfied if the residual functions
\begin{flalign}
\label{DI_Fred_1}
& \delta_1(x,y) \!=\! R(x,y)\!-\!P(x,y)\!-\!\int_0^y P(x,s) S(s,y) ds \nonumber \\[-0.1cm]
& \!-\!\int_y^x P(x,s)R(s,y) ds \!-\! \int_x^1 Q(x,s)R(s,y)dy, \forall (x,y) \in \Omega_L,\\[-0.4cm]
\label{DI_Fred_2}
& \delta_2(x,y) \!=\! S(x,y)\!-\!Q(x,y)\!-\!\int_0^xP(x,s)S(s,y) ds \nonumber \\[-0.1cm]
& \!-\!\int_x^y Q(x,s) S(s,y) ds \!-\! \int_y^1Q(x,s)R(s,y)ds, \forall (x,y) \in \Omega_U,
\end{flalign}

\noindent
are identically zero in their respective triangular domains. Since \eqref{DI_Fred_0} does not depend on any original and target systems, it can be used to find an approximation of the inverse Kernels $R$ and $S$, given the approximate direct ones $P\approx N$, $Q\approx M$.

\begin{proposition}
\label{Prop_6} Let $A(x,y)=\sum_{k=0}^{z(d)-1} a_k x^{\alpha_k} y^{\beta_k}$ and $B(x,y)=\sum_{k=0}^{z(d)-1} b_k x^{\alpha_k} y^{\beta_k}$ be the polynomial approximations of $R$ and $S$, respectively, of arbitrary even degree $d \in \mathbb{N}$, with coefficients $a_k$ and $b_k$ and powers in accordance with \eqref{Def_pol}. Let $\delta_1=\delta_1(x,y)$ and $\delta_2=\delta_2(x,y)$ be the resulting residual functions according to \eqref{DI_Fred_1} and \eqref{DI_Fred_2}, respectively, with degree  $d_{\delta}=2\left\lceil(\max\{d_N,d_M\}+d+1)/2\right\rceil$ and $\gamma_j \geq 0, j=1,\ldots,4$. For given direct approximate Kernels $N$ and $M$ of $P$ and $Q$ with degrees $d_N$ and $d_M$, respectively (solution of \eqref{PIDE_opt_1}-\eqref{PIDE_opt_11}), the integral equation \eqref{DI_Fred_0}-\eqref{DI_Fred_2} can be formulated as the convex optimization problem:

\noindent
\begin{flalign}
\label{IE_opt_1}
& \underset{\gamma_j, A, B, s_j}{\text{minimize:}} \qquad \gamma_1+ \gamma_2 + \gamma_3 + \gamma_4 \\
&\text{subject to:} \nonumber \\
\label{IE_opt_4}
& \left(\delta_1(x,y)-\underline{\rho}_1-[s_1(x,y) \ s_2(x,y)] \ g_L(x,y)\right) \in \Sigma_s,\\
\label{IE_opt_5}
& \left(\overline{\rho}_1-\delta_1(x,y)-[s_3(x,y) \ s_4(x,y)] \ g_L(x,y)\right) \in \Sigma_s, \\
\label{IE_opt_6}
& \left(\delta_2(x,y)-\underline{\rho}_2-[s_5(x,y) \ s_6(x,y)] \ g_U(x,y) \right) \in \Sigma_s, \\
\label{IE_opt_7}
& \left(\overline{\rho}_2-\delta_2(x,y)-[s_7(x,y) \ s_8(x,y)] \ g_U(x,y) \right) \in \Sigma_s, \\
\label{IE_opt_8}
& s_1,s_2,s_3,s_4,s_5,s_6,s_7,s_8 \in \Sigma_s,\\
\label{IE_opt_9}
& \begin{bmatrix} \ \gamma_1 & \underline{\rho}_1 \\ \ \underline{\rho}_1 & \gamma_1 \end{bmatrix} \succeq 0, \quad  \begin{bmatrix} \ \gamma_2 & \overline{\rho}_1 \\ \ \overline{\rho}_1 & \gamma_2 \end{bmatrix} \succeq 0,\\
\label{IE_opt_10}
& \begin{bmatrix} \ \gamma_3 & \underline{\rho}_2 \\ \ \underline{\rho}_2 & \gamma_3 \end{bmatrix} \succeq 0, \quad  \begin{bmatrix} \ \gamma_4 & \overline{\rho}_2 \\ \ \overline{\rho}_2 & \gamma_4 \end{bmatrix} \succeq 0,\\
\label{IE_opt_11}
& \delta_1=\delta_1(x,y) \left|_{R \approx A, \ S \approx B}^{P \approx N, \ Q \approx M} \right. \ \mbox{as in (\ref{DI_Fred_1})}, \\
\label{IE_opt_12}
& \delta_2=\delta_2(x,y) \left|_{R \approx A, \ S \approx B}^{P \approx N, \ Q \approx M} \right. \ \mbox{as in (\ref{DI_Fred_2})},
\end{flalign}
for some polynomials $s_j, j=1,\ldots,8$ of degree $d_{\delta}-2$, $g_L=[g_1, g_2]$, $g_U=[g_1, g_3]$, $g_1(x)=x(1-x)$, $g_2(x,y)=y(x-y)$ and $g_3(x,y)=(1-y)(y-x)$. The optimal minimal bounds for the residual functions are:  $\overline{\delta_1}=\max\{\gamma_1,\gamma_2\}$ and $\overline{\delta_2}=\max\{\gamma_3,\gamma_4\}$.
\end{proposition}

\begin{proof}
The proof follows the same arguments of the one of Proposition \ref{Prop_2}, hence it is omitted.
\end{proof}

\section{Numerical Results}
\label{Results}

\noindent
The numerical solution of the convex optimization problems proposed in this article has been obtained via the Yalmip toolbox for Matlab \cite{L-2004} using the SDP package part of the Mosek solver \cite{Mosek-2016}.

\subsection{Parabolic PDE with constant reactivity term}

This example considers $\lambda=20$ and $\epsilon=1$ in the system \eqref{SP_1} and $c=0$ in the target system \eqref{TP_1}. The direct kernel $K$ in (\ref{Vol_Op}) is approximated solving the convex optimization problem \eqref{PDE_opt_1}-\eqref{PDE_opt_13}. The bounds of the residual functions $\delta_1$ and $\delta_2$ are indicated in Figure \ref{figure_0}. The approximate kernel $N$ for $d=12$ is depicted in Figure \ref{figure_1}(a). Regarding the inverse transformation \eqref{Inv_Vol_Op}, using the previous direct approximate kernel, a moment sequence \eqref{L_mom} has been calculated solving the convex optimization problem \eqref{Mom_opt_1}-\eqref{Mom_opt_15} for $d_m=13$. Table \ref{table_2} shows a small part of this sequence, which is compared with the inverse kernel ${L}^\star$ obtained via the successive approximation method, the positive component of which
is given by $\breve{L}^\star(x,y)=-\overline{\lambda} y J_{1}(\sqrt{\Theta})/\sqrt{\Theta}+10$, with $ \Theta=\overline{\lambda}(x^2-y^2)$, where $J_1$ is the first order Bessel function \cite{SK-2010}. The whole moment sequence has at least 2 digits of precision with respect to $\breve{L}^\star$. Finally, based on the whole previous moment sequence, an approximate inverse kernel $L_r=\breve{L}_r-\sigma$ has been determined solving  \eqref{Mom_rec_1}-\eqref{Mom_rec_5} for $r=12$, the result of which is shown in Figure \ref{figure_1}(b).

\begin{figure}[thpb]
\centering
\vskip-0.4cm
\hskip-0.25cm
      \begin{tabular}{c}
       \includegraphics[width=8.5cm,height=6.5cm]{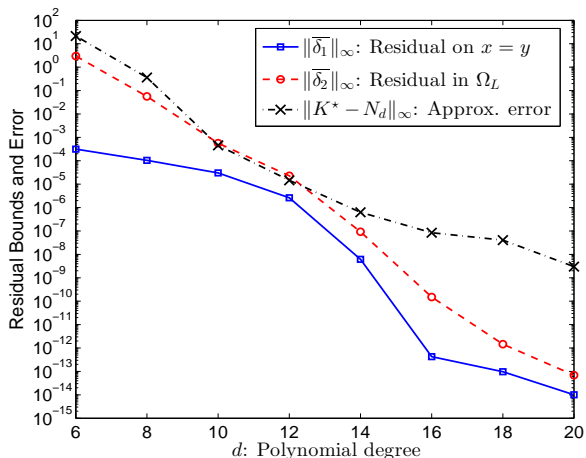}
       \end{tabular}
        \vspace{-0.25cm}
       \caption{Bounds for the residual functions in \eqref{P1_T}, solution of \eqref{PDE_opt_1}-\eqref{PDE_opt_13}, as a function of polynomial degrees $d$.}
       \label{figure_0}
   \end{figure}

\begin{table}[thpb]
\centering
\begin{tabular}{ c | c | c | c | c | c | c | c }
\hline $m_{i,j}$ &  $m_{0,0}$ & $m_{0,1}$ & $m_{0,2}$ & $m_{0,3}$ & $m_{0,4}$ & $m_{0,5}$ & $\sigma$ \\
\hline \hline
$\breve{L}^\star$ & $4.19026$ & $1.22856$ & $0.55171$ & $0.30114$ & $0.18434$ & $0.12179$ & $10.00000$ \bigstrut \\
\hline
$\breve{L}$ & $4.19155$ &$1.22899$ & $0.55192$ & $0.30127$  & $0.18443$ & $0.12185$ & $10.00258$ \bigstrut  \\
\hline
\end{tabular}
\vskip+0.2cm
\caption{Extract of the moment sequence \eqref{L_mom}, optimal solution of \eqref{Mom_opt_1}-\eqref{Mom_opt_15}, for $r=0,\ldots,13$.}
\label{table_2}
\end{table}

\begin{figure}[thpb]
\centering
\vskip-0.2cm
      \begin{tabular}{c}
       \includegraphics[width=7.5cm,height=6.5cm]{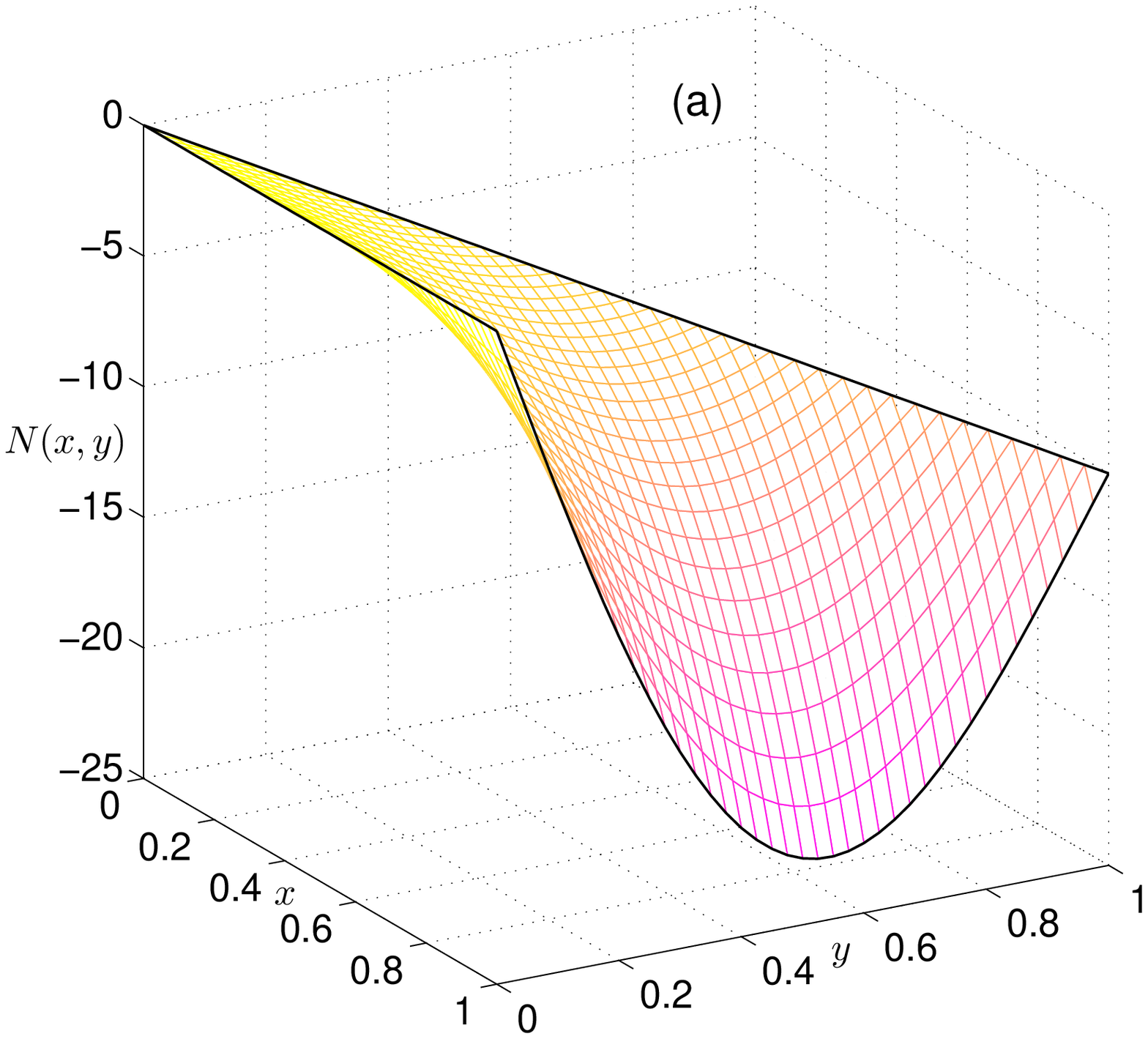}\\[-0.58cm]
       \includegraphics[width=7.5cm,height=6.5cm]{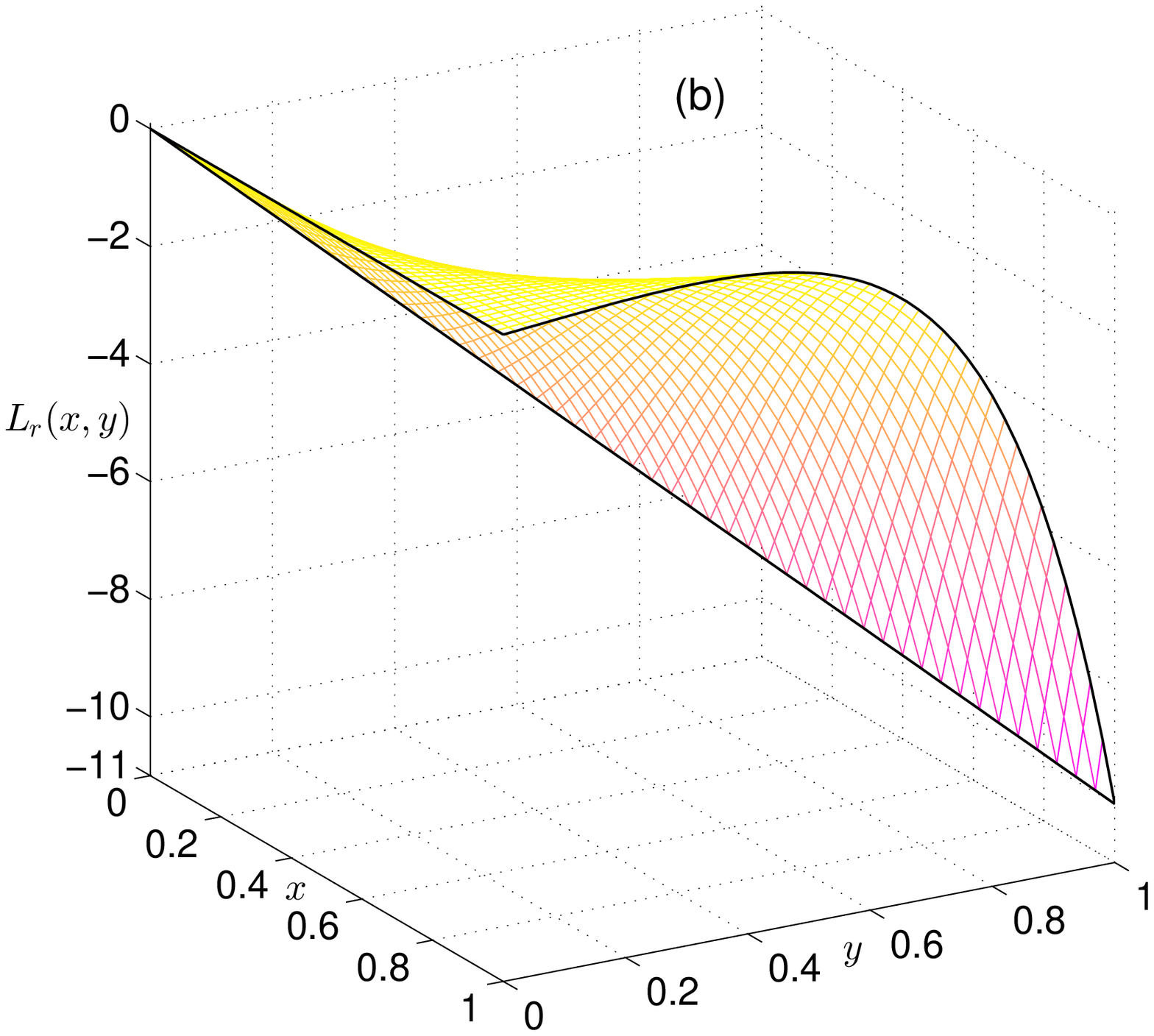}
       \end{tabular}
        \vspace{-0.38cm}
       \caption{(a) Direct approximate kernel $N$ for $K$ in \eqref{Vol_Op}, $d=12$. (b) Inverse approximate kernel $L_r$ for $L$ in \eqref{Inv_Vol_Op}, $r=12$.}
       \label{figure_1}
       \vskip-0.1cm
   \end{figure}

\subsection{Hyperbolic PIDE: Fredholm-type Operator}
This example considers the problem presented in \cite{BK-2014}\cite{BK-2015}:
\begin{flalign}
\label{fun_1}
f(x) &= a + \frac{b \ \sigma}{\sqrt{c} \ \cosh(\sqrt{c})}  \ \sinh\left(\sqrt{c}(1-x)\right) \\
\label{fun_2}
h_2(x) &= - \frac{b \ \sigma}{\cosh(\sqrt{c})} \ \cosh(\sqrt{c}x) \ \cosh\left(\sqrt{c}(1-y)\right) \\
\label{fun_3}
h_1(x) &= h_2(x) + b \ \sigma \cosh\left(\sqrt{c}(x-y)\right)
\end{flalign}
with $a=1.25, \ b=0.1, \ c=0.1, \ \sigma=10$ (the application of the proposed approach is not limited to this case, which has been selected for comparison purposes). This problem (equations (71)-(73) of \cite{BK-2015}) corresponds to a first-order PDE coupled with a second order ODE, equivalent to the 1-dimensional hyperbolic PDE \eqref{SH_1}. The direct Kernels $P$ and $Q$ in \eqref{Fred_Op} have been approximated solving the convex optimization problem \eqref{PIDE_opt_1}-\eqref{PIDE_opt_11}. To implement this approach, the functions $f$, $h_1$ and $h_2$ in \eqref{fun_1}-\eqref{fun_3} have been approximated by a combination of univariate polynomials of degree $4$, with maximum approximation error $<4.2\!\cdot 10^{-8}$. The approximate Kernels $M$ and $N$ for a polynomial degree $d=10$ are shown in Figure \ref{figure_2}(a), with root mean square bounds for the residual functions: $\gamma_1=4.70\!\cdot\!10^{-10}$ and $\gamma_2=1.07\!\cdot\!10^{-9}$. Using the previous result, the inverse Kernels $R$ and $S$ have been approximated solving (\ref{IE_opt_1})-(\ref{IE_opt_12}) for a polynomial degree $d=10$. The result is shown in Figure \ref{figure_2}(b) with bounds of the residual functions: $\overline{\gamma}=\max\{\gamma_1,\gamma_2,\gamma_3,\gamma_4\}\leq 8.01\!\cdot\!10^{-10}$.

\begin{figure}[thpb]
\centering
\vskip-0.45cm
\hskip-0.2cm
      \begin{tabular}{c}
       \includegraphics[width=8cm,height=6.7cm]{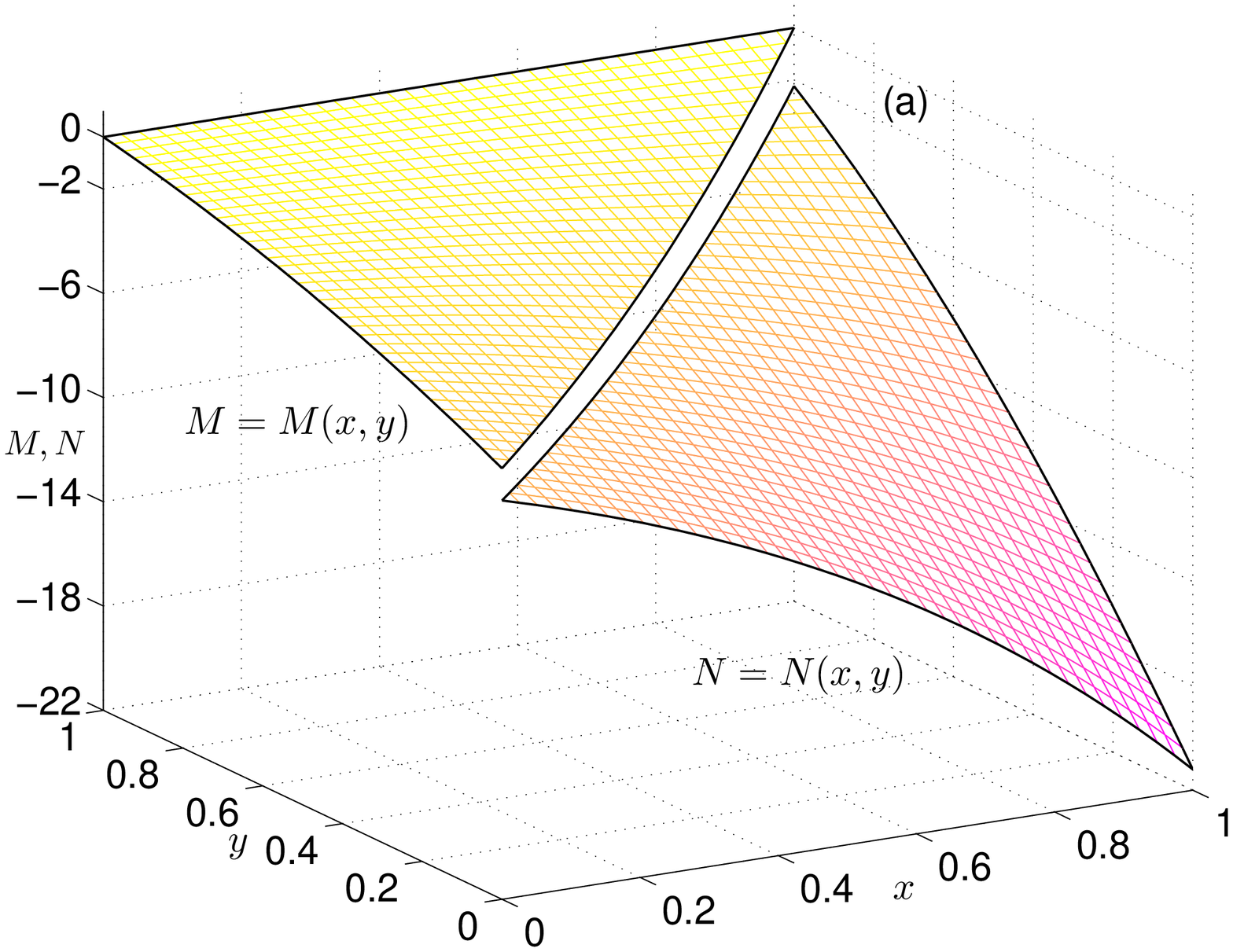}\\[-0.55cm]
       \includegraphics[width=8cm,height=6.7cm]{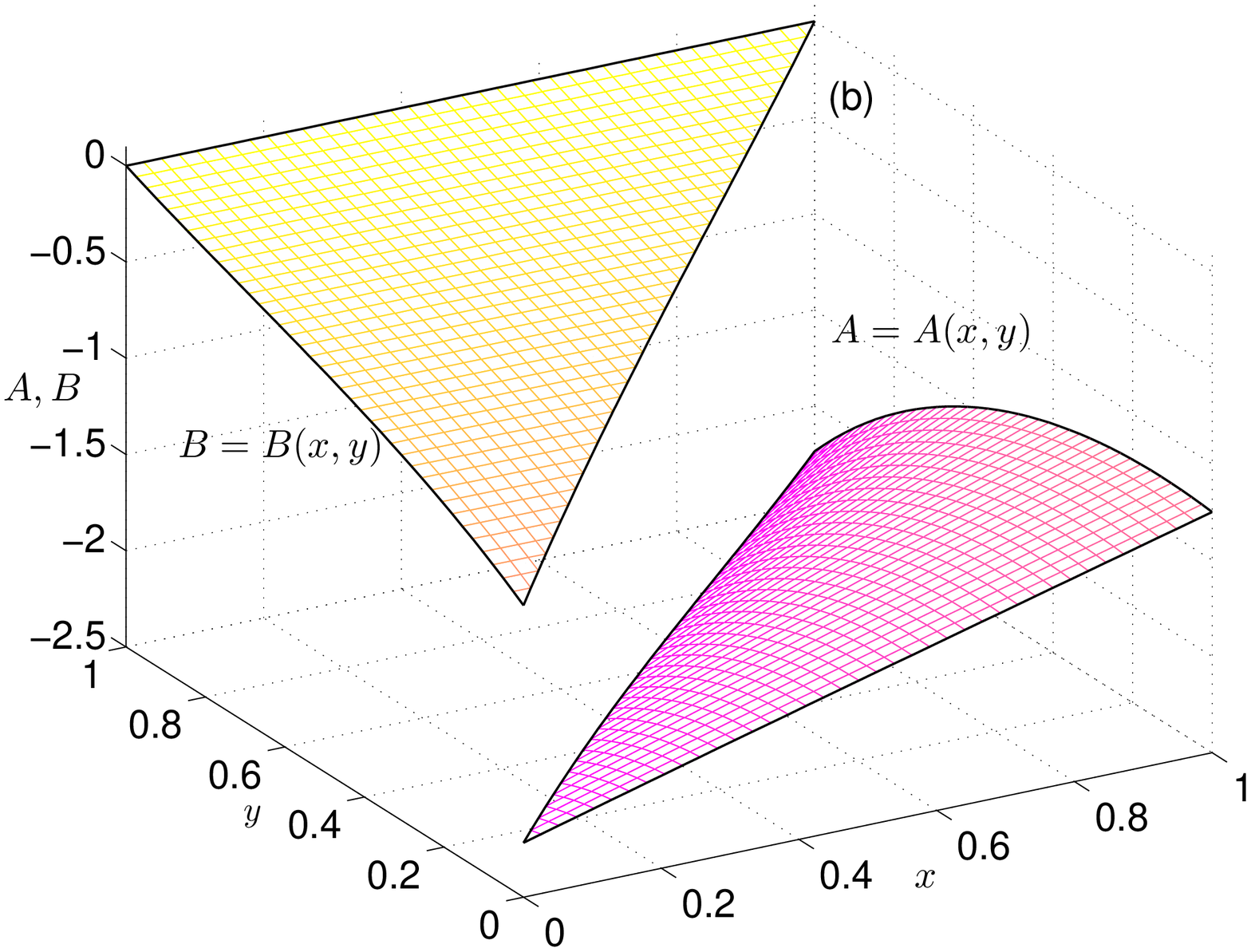}
       \end{tabular}
        \vspace{-0.3cm}
       \caption{(a) Direct approximate Kernels $M$ and $N$ for $d=10$ ($P \approx N$, $Q \approx M$ in \eqref{Fred_Op}). (b) Inverse approximate Kernels $A$ and $B$ for $d=10$ ($R \approx A$, $S \approx B$ in \eqref{Inv_Fred_Op}).}
       \label{figure_2}
       \vskip-0.1cm
   \end{figure}

\section{Conclusions}
\label{Conclu}
In this paper a convex optimization approach to Backstepping PDE design for systems with strict and non-strict feedback structure, involving Volterra and Fredholm operators has been presented. The approach proposed allows obtaining approximate solutions with sufficient precision to guarantee the stability of the system in the $\mathcal{L}^2$-norm topology. The numerical examples illustrate the performance of the approach proposed and the flexibility of SOS-convex optimization to manage problems with operators of different structure and objectives. The method is restricted to systems involving functions which can be approximated by polynomials with computationally tractable degree and Kernels continuous or piece-wise continuous. The main limitation of this method is the current state of Sum-of Squares tools, regarding the type of monomials used in the decompositions, and the convex optimization tools, in relation with managing a large number of parameters and parameters with big magnitudes.
%



\end{document}